\documentclass[11pt]{article}
\usepackage{verbatim} %
\usepackage[usenames, dvipsnames]{color}

\usepackage{bbm}
\usepackage{mathtools}
\usepackage{amsmath}
\usepackage{amsthm, thmtools}
\usepackage{thm-restate} %
\usepackage{amstext,amssymb, amsfonts}

\usepackage{dsfont} %
 
\usepackage{empheq} %

\usepackage{lmodern} %

\usepackage{float} %

\usepackage{array,multirow} %

\usepackage{algorithmic}
\usepackage[section]{algorithm} %
\usepackage{todonotes}

\newcounter{mynotes}
\usepackage{hyperref} %
 
\usepackage[symbol]{footmisc}

\declaretheorem[within=section]{theorem}
\declaretheorem[sibling=theorem]{corollary}

\declaretheorem[sibling=theorem]{lemma}
\declaretheorem[sibling=theorem]{claim}

\declaretheorem[sibling=theorem]{definition}

\declaretheorem[sibling=theorem]{proposition}
\declaretheorem[sibling=theorem]{remark}
\declaretheorem[sibling=theorem]{question}
\declaretheorem[sibling=theorem]{conjecture}

\def\tV{{\widetilde{V}}}
\def\tA{{\widetilde{A}}}
\def\tW{{\widetilde{W}}}
\def\tX{{\widetilde{X}}}
\def\tQ{{\widetilde{Q}}}
\def\tV{{\widetilde{V}}}

\def \barb{{\bar{b}}}
\def \barE{{\overline{E}}}
\def \barS{{\overline{S}}}
\def \barT{{\overline{T}}}
\def \barA{{\overline{A}}}

\newcommand{\R}{\mathbb{R}} %
\newcommand{\C}{\mathbb{C}} %
\newcommand{\F}{\mathbb{F}}

\newcommand{\cP}{\mathcal P}

\newcommand{\supp}{\mathrm{supp}} %
\newcommand{\inpro}[2]{\left\langle #1,#2 \right\rangle} %
\newcommand{\E}{\mathbb{E}} %
\newcommand{\set}[1]{\{#1\}}

\renewcommand{\epsilon}{\varepsilon}

  \newcommand{\beq}{\begin{equation}}
  \newcommand{\eeq}{\end{equation}}
  \newcommand{\beqn}{\begin{equation*}}
  \newcommand{\eeqn}{\end{equation*}}
  \newcommand{\beqr}{\begin{eqnarray}}
  \newcommand{\eeqr}{\end{eqnarray}}
  \newcommand{\beqrn}{\begin{eqnarray*}}
  \newcommand{\eeqrn}{\end{eqnarray*}}
  \newcommand{\bmline}{\begin{multline}}
  \newcommand{\emline}{\end{multline}}
  \newcommand{\bmlinen}{\begin{multline*}}
  \newcommand{\emlinen}{\end{multline*}}

\usepackage[margin=1in]{geometry}
\usepackage{url}

\renewcommand{\le}{\leqslant}

\renewcommand{\ge}{\geqslant}

\newcommand{\MDS}{\operatorname{MDS}}
\newcommand{\Span}{\operatorname{span}}
\newcommand{\rank}{\operatorname{rank}}
\newcommand{\row}{\mathrm{row}}
\newcommand{\col}{\mathrm{col}}

\newcommand{\lp}{\left(}
\newcommand{\rp}{\right)}

\newcommand{\lc}{\left\{}
\newcommand{\rc}{\right\}}

\title{Lower Bounds for Maximally Recoverable Tensor Codes\\ and Higher Order MDS Codes}

\author{Joshua Brakensiek\thanks{Department of Computer Science, Stanford University, Stanford, CA. Email: {\tt jbrakens@cs.stanford.edu}. Portions of this work were done during an internship at Microsoft Research, Redmond. Research supported in part by an NSF Graduate Research Fellowship.} \and Sivakanth Gopi\thanks{Microsoft Research, Redmond, WA. Email: {\tt sigopi@microsoft.com}.} \and Visu Makam\thanks{Radix Trading Europe B.V. Email: {\tt
    visu@umich.edu}. Research supported by NSF Grant No. DMS-1638352,
  CCF-1412958, and CCF-1900460 and the University of Melbourne and the
  Institute for Advanced Study, Princeton.}}

\date{}

\begin{document}

\maketitle
\thispagestyle{empty}

\begin{abstract}
An $(m,n,a,b)$-tensor code consists of $m\times n$ matrices whose columns satisfy `$a$' parity checks and rows satisfy `$b$' parity checks (i.e., a tensor code is the tensor product of a column code and row code). Tensor codes are useful in distributed storage because a single erasure can be corrected quickly either by reading its row or column. Maximally Recoverable (MR) Tensor Codes, introduced by Gopalan et al.~\cite{GHKSWY}, are tensor codes which can correct every erasure pattern that is information theoretically possible to correct. The main questions about MR Tensor Codes are characterizing which erasure patterns are correctable and obtaining explicit constructions over small fields.

In this paper, we study the important special case when $a=1$, i.e., the columns satisfy a single parity check equation. We introduce the notion of higher order MDS codes ($\MDS(\ell)$ codes) which is an interesting generalization of the well-known MDS codes, where $\ell$ captures the order of genericity of points in a low-dimensional space. We then prove that a tensor code with $a=1$ is MR iff the row code is an $\MDS(m)$ code. We then show that $\MDS(m)$ codes satisfy some weak duality. Using this characterization and duality, we prove that $(m,n,a=1,b)$-MR tensor codes require fields of size $q=\Omega_{m,b}(n^{\min\{b,m\}-1})$. Our lower bound also extends to the setting of $a>1$. We also give a deterministic polynomial time algorithm to check if a given erasure pattern is correctable by the MR tensor code (when $a=1$).
\end{abstract}

\vspace{-3ex}

\newpage
\tableofcontents
\thispagestyle{empty}
\newpage

\pagenumbering{arabic}
\section{Introduction}

In distributed storage, data is stored in individual servers each with a few terabytes of capacity. A large datacenter can have millions of such servers holding a few exabytes (millions terabytes) of data. Costs for building and running such datacenters run into billions of dollars. In such a large system, hard disks crash every minute. Servers can also become temporarily unavailable due to system updates or network bottlenecks. To avoid data loss and to serve user requests with low latency, some form {of} redundancy is necessary. Replicating data multiple times is too wasteful, doubling or tripling the costs. Erasure coding has been used to improve the storage efficiency while maintaining data reliability. For example a $(k+h,k)$-Reed-Solomon code can be used to add $h$ redundant servers to every $k$ data servers. This allows us to correct any $h$ erasures (node failures). To improve storage efficiency, one is forced to choose large values of $k$. But this creates problems with latency, to recover a single erased node, one needs to read data from $k$ other servers. When $k$ is large, this is prohibitively slow. To balance this tension between storage efficiency and latency, erasure codes with \emph{locality} were introduced in~\cite{GHSY,XOR_ELE}. These codes allow fast recovery of an erased symbol by reading a small number of other unerased coordinates, this ability is referred to as locality. This requires that each coordinate of the code participates in a parity check equation involving a few other symbols. In addition to these local parity checks, the code also satisfies a small number of global parity checks which give it resilience to tolerate a large number of erasures in the worst case. Such codes with different architectures were deployed in practice to reduce the storage overhead while maintaining low latency and high durability~\cite{HuangSX,MLRH14}. 

The key notion of \emph{maximal recoverability} for such local codes
was introduced in \cite{CHL,GHSY}. Maximal recoverability refers to
the optimality of a code {with respect} to its ability to correct every erasure pattern that is information theoretically possible to correct given the code architecture (or topology). Therefore maximally recoverable (MR) codes have the best durability among all the codes with that particular architecture. In a seminal paper, Gopalan et al.~\cite{GHKSWY} generalized and brought under a common framework various code topologies used in erasure coding, by introducing MR codes with grid-like topologies. Here the topology {is} specified by $T_{m\times n}(a,b,h)$. This means that the codewords are $m\times n$ matrices where each column satisfies $a$ parity check equations and each row satisfies $b$ parity check equations. In addition there are $h$ global parity check equations that all the $mn$ symbols satisfy. It is easy to show the existence of MR codes for any given topology over exponentially large fields by using randomization and Schwartz-Zippel lemma~\cite{GHKSWY}. There are two main questions about these codes which are still wide open:
\begin{question}
	\label{question:erasure_patterns_abh}
	 What are the erasure patterns that are correctable by an MR code with topology $T_{m\times n}(a,b,h)$?
\end{question}

\begin{question}
	\label{question:fieldsize_abh}
	 What is the minimum field size required to construct an MR code with topology $T_{m\times n}(a,b,h)$? In particular, can we get explicit constructions over small fields?
\end{question}

Both questions are really important. Knowing which patterns are correctable allows to design the topology which gives the desired durability while minimizing the storage costs. And explicit constructions over small fields are important for them to be useful in practice. The amount of computation needed for encoding and recovering from erasures is very sensitive to the field size over which the code is defined. Typically, field sizes of $2^8$ to $2^{16}$ are used in practice. Fields which are much bigger incur a large computational overhead and therefore infeasible to use in practice.

When $h=0$, codes with topology $T_{m\times n}(a,b,0)$ are pure tensor
codes {(also called product codes)}, i.e., the code is a tensor product of a column code and a row code. In this paper, we will denote a code with this topology as an $(m,n,a,b)$-tensor code. Let us denote the row code by denoted by $C_\row$, which is a $(n,n-b)$-code. And denote the column code by $C_\col$, which is a $(m,m-a)$-code. Then the tensor code is $C=C_\col \otimes C_\row$, the codewords are $m\times n$ matrices where each row belongs to $C_\row$ and each column belongs to $C_\col.$ The setting $h=0$ is already very interesting for the following reason. In a recent work~\cite{holzbaur2021correctable}, it was shown that the set of erasure patterns correctable by an MR code with topology $T_{m\times n}(a,b,h)$ are precisely those obtained by adding $h$ more erasures arbitrarily to erasure patterns correctable by an MR code with topology $T_{m\times n}(a,b,0)$. Therefore answering Question~\ref{question:erasure_patterns_abh} for $h=0$ is enough to answering it for any $h$. Moreover, if one can construct an explicit MR code with topology $T_{m\times n}(a,b,0)$ over $\F_q$, then one can get an explicit MR code with topology $T_{m\times n}(a,b,{h})$ over fields of size $q^{(m-a)(n-b)}$~\cite{holzbaur2021correctable}, which also partially answers Question~\ref{question:fieldsize_abh}. 

In this paper, we will focus on $(m,n,a,b)$-MR tensor codes in the special case of $a=1$, i.e., there are no global parity checks and all the columns satisfy a single parity check equation. Firstly, this setting is very much practically relevant. The f4 storage architecture of Facebook~\cite{MLRH14} uses an $(m=3,n=14,a=1,b=4)$-tensor code, though they couldn't obtain an MR construction in their implementation. They simply use a tensor product of $(14,10)$-Reed-Solomon code with a $(3,2)$-parity check code, which need not be MR and therefore doesn't have the optimal durability. Moreover, as we will see, constructing MR tensor codes even in this special case of $a=1$ is quite challenging and leads to some really interesting generalization of MDS (Maximum Distance Separable) codes.

\subsection{Previous work}
\paragraph{Correctable Patterns:} In the paper where they introduce MR tensor codes, \cite{GHKSWY} characterize the set of correctable erasure patterns by an $(m,n,a,b)$-MR tensor code when $a=1$ in terms of a combinatorial condition called `regularity'. But they don't give {an} efficient procedure to check if an erasure pattern is regular, naively it would require checking an exponential number of constraints.  They also conjectured that regularity characterizes correctable erasure patterns when $a>1$. But this conjecture is false as shown in~\cite{holzbaur2021correctable}, we will later present a counterexample with an illuminating explanation of why regularity fails to capture correctability when $a>1.$ Currently, we do not have any characterization of correctable erasure patterns by an $(m,n,a,b)$-MR tensor codes when $a,b>1$. A subset of correctable erasure patterns in MR tensor codes when $a=2$ were obtained in~\cite{Shivakrishna_MRproduct}.

\paragraph{Constructions:} If we instantiate the row code and column code with random codes over fields of size $q\gg (an+bm-ab)\cdot \binom{mn}{an+bm-ab}$, by Schwartz-Zippel lemma and union bound, we can conclude that with high probability the tensor code will be MR. By doing a more careful union bound, \cite{kong2021new} show that there exists $(m,n,a=1,b)$-MR tensor codes over fields of size $q=O_{m,b}\lp n^{b(m-1)}\rp.$ In some special cases, $(m=4,n,a=1,b=2)$ and $(m=3,n,a=1,b=3)$, \cite{kong2021new} proved the existence of MR tensor codes over fields of size $q=O(n^5).$

\paragraph{Lower bounds:} Prior to our work there are no general lower bounds on the field size required for MR tensor codes. In the special case of $(m=4,n,a=1,b=2)$-MR tensor code, \cite{kong2021new} prove a quadratic lower bound on the field size, i.e., $q=\Omega(n^2).$ For codes with topology given by $T_{n\times n}(a=1,b=1,h=1),$ \cite{GHKSWY} prove a lower bound of $\exp\lp \Omega\lp\log(n)^2\rp\rp$. This was improved by \cite{kane2019independence}, where they proved the optimal field size is $q=\exp(O(n)).$

{
\paragraph{MR Local Reconstruction Codes} MR codes with topology $T_{m\times n}(a,0,h)$ are called MR Local Reconstruction Codes (MR LRCs).\footnote{MR LRCs are also called Partial-MDS (Maximum Distance Separable) codes in prior works.} There is extensive body of work on MR LRCs. Several explicit constructions of MR LRCs over small fields are given in~\cite{GHJY,GYBS,GLX-ff,MK19,gopi2020maximally,Blaum,TPD,HY,GHKSWY,CK,BPSY,MK19,GLX-ff,cai2020construction,UMP-2020}. A construction of an MR LRC with field size $$q \le \lp O\lp \max\{m,n\}\rp^{\min\{h,m-a\}}\rp$$ is given in~\cite{gopi2020improved,cai2020construction} using skew polynomials. See the prior work section in~\cite{gopi2020improved} for a survey of existing results on MR LRCs. A lower bound of $$q \gtrsim_{h,a}\left(n m^{\min\{a+1,h-1\}}\right)$$ on the field size is shown in~\cite{gopi2020maximally}. The upper and lower bounds match for the setting where $n=m, $ and $\alpha m \ge a \ge h-2$ for some constant $\alpha <1$, showing that the optimal field size is $\Theta_h(n^h)$ in this case. Closing the gap between upper and lower bounds for general setting of parameters is a major open problem.
}

\subsection{Our Contributions}
We show an equivalent characterization of MR tensor codes with $a=1$ in terms of their row codes. In particular, we introduce the notion of higher order MDS codes, which is a natural generalization of MDS (Maximum Distance Separable) codes and prove a tight relation to MR tensor codes.

\paragraph{Higher order MDS codes}: A $(n,k)$-MDS code encodes $k$
symbols into $n$ symbols, such that one can correct any ${k}$
erasures. Moreover, Reed-Solomon codes are explicit constructions of
MDS codes over fields of size $O(n)$. MDS codes play crucial role in
coding theory and especially in erasure coding for distributed
storage. Suppose the generator matrix of an $(n,k)$-code $C$ over $\F$
is given by a $k\times n$ matrix V, whose columns are denoted by
$V_1,V_2,\dots,V_n\in \F^k.$ Let $V_A=\operatorname{span}\lc V_i: i\in
A\rc$.\footnote{{In some instances, $V_A$ will denote the
    matrix whose rows are the $V_i$ with $i \in A$.}} Then $C$ is MDS iff $V_A=\F^k$ for all $A\subset[n]$ of size $|A|=k$. An equivalent way to state this is to say that for any subsets $A_1,A_2\subset [n]$, $\dim(V_{A_1} \cap V_{A_2}) = \dim(W_{A_1}\cap W_{A_2})$ for some generic matrix $W$.\footnote{\textbf{Genericity:} A generic point $X$ can be thought of either as a symbolic vector, or one can think of it as a point with entries in an infinite field $\F$ which avoids any fixed low-dimensional algebraic variety. If $\F=\R$ or $\C$, then one can think of a generic point as something which escapes any measure zero set. In particular, low-dimensional varieties are measure zero sets.}
One direction is obvious, this new condition clearly implies the usual
MDS definition by taking $A_2=[n]$. To prove the other direction note
that\footnote{{Here, the sum of two vectors spaces is $V_1 + V_2 :=
  \operatorname{span}\{v_1 + v_2: v_1 \in V_1, v_2 \in V_2\}.$}}
\begin{align*}
\dim(W_{A_1}\cap W_{A_2})&=\dim(W_{A_1})+\dim(W_{A_2})-\dim(W_{A_1} + W_{A_2})\\
&=\dim(W_{A_1})+\dim(W_{A_2})-\dim(W_{A_1\cup A_2})\\
&=\min\{|A_1|,k\}+\min\{|A_2|,k\}-\min\{|A_1\cup A_2|,k\}\tag{Genericity of $W$}\\
&=\dim(V_{A_1})+\dim(V_{A_2})-\dim(V_{A_1\cup A_2})\tag{$V$ is MDS}\\
&=\dim(V_{A_1})+\dim(V_{A_2})-\dim(V_{A_1}+V_{A_2})\\
&=\dim(V_{A_1}\cap V_{A_2}).
\end{align*}
This leads to a natural generalization of MDS codes to higher order MDS codes.

\begin{definition}[Higher order MDS code ($\MDS(\ell)$)]\label{def:higherMDS}
  Let $C$ be an $(n,k)$-code with generator matrix $V_{k\times n}.$
  For $\ell \ge 2$, we say that $C$ is an $\MDS(\ell)$ code if for all $A_1, \hdots, A_{\ell} \subset [n]$,
  \begin{equation}
  	\label{eqn:MDSell}
    \dim(V_{A_1} \cap V_{A_2} \cap \cdots \cap V_{A_\ell})=\dim(W_{A_1} \cap \cdots \cap W_{A_{\ell}}),
  \end{equation}
  where $W$ is a $k\times n$ generic matrix. %
\end{definition}

Note that the definition of $\MDS(\ell)$ is independent of the generator chosen to represent the code $C$ and therefore purely a property of the code $C$. This is because of invariance of the condition (\ref{eqn:MDSell}) under basis change $V \to A\cdot V$ for any invertible $k\times k$ matrix $A.$

\paragraph{Example - $\MDS(3)$}: It is instructive to look at an example. Suppose $C$ is an $(n,3)$-$\MDS(3)$ code over $\F$. Let $V$ be its generator matrix of size $3 \times n$ and let $V_1,V_2,\dots,V_n\in \F^3$ be its columns. By abusing notation, we can think of $V_i$ as points in the two dimensional projective space $\mathbb{P}^2(\F)$, because scaling the vectors does affect the definition of $\MDS(3).$ Therefore the subspace $\Span\{V_i,V_j\}$ corresponds to a line passing through the points $V_i,V_j\in \mathbb{P}^2(\F).$ Now the usual notion of MDS (which is equivalent to $\MDS(2)$), corresponds to the condition that no three points among ${V_1,V_2,\dots,V_n}\subset \mathbb{P}^2(\F)$ are collinear. The $\MDS(3)$ condition corresponds to the condition that no three points are collinear and if we draw all the lines through pairs of points in ${V_1,V_2,\dots,V_n}$, then no three lines are concurrent (i.e., pass through the same point) other than the trivial concurrency which occurs when the the three lines chosen pass through some $V_i$. Thus $\MDS(3)$ can be thought of a higher order genericity condition than $\MDS(2)$ and therefore strictly stronger condition than $\MDS(2)$. In particular an arbitrary MDS code may not be $\MDS(3).$ More generally, $\MDS(\ell)$ require that $\ell$-wise intersections behave generically. Thus $\ell$ can be thought of as the degree of genericity of the points $V_1,V_2,\dots,V_n.$

We now state one of our main theorems relating MR tensor codes with higher order MDS codes.
\begin{theorem}
	\label{thm:MRtensorcodesa1_MDSm}
	Let $C=C_\col \otimes C_\row$ be an $(m,n,a=1,b)$-tensor code where $C_\col$ is a simple parity check code. Then $C$ is MR iff $C_\row$ is an $\MDS(m)$ code.
\end{theorem}
Therefore constructing MR tensor codes when $a=1$ is equivalent to constructing higher order MDS codes. Our next result is a lower bound on the field size required for higher order MDS codes. We will be stating it in the regime when the codimension of the code is a constant, which is the regime of interest in practice.

\begin{theorem}\label{thm:main-lower-bound}
	Let $C$ be an $(n,k)$-$\MDS(\ell)$ code over $\F_q$. Then $$q \ge \Omega_{\ell}\lp n^{\min\{\ell,k,n-k\}-1}\rp.$$
\end{theorem}

This immediately implies the following corollary for MR tensor codes.
\begin{corollary}\label{cor:main-lower-bound}
	Let $C$ be an $(m,n,a,b)$-MR tensor code over $\F_q$. Then $$q \ge \Omega_{m}\lp n^{\min\{m-a+1,b,n-b\}-1} \rp.$$
\end{corollary}
This is the first general lower bound on the field size required for MR tensor codes. Prior to our work, only a quadratic lower bound was known in the special case of $(m=4,n,a=1,b=2)$~\cite{kong2021new}, but unfortunately we cannot recover this bound from our general result.

We also show the following upper bound on field size for MR tensor codes, generalizing \cite{kong2021new} where they obtained such a result for $a=1$.
\begin{theorem}\label{thm:upper}
	There exist $(m,n,a,b)$-MR tensor codes over fields of size $$q=O_{m,b}\lp n^{b(m-a)}\rp.$$
\end{theorem}

Finally, we also give an efficient polynomial time algorithm for checking if an erasure pattern is correctable by an MR tensor code when $a=1$.
\begin{theorem}\label{thm:fast}
	There exists an efficient algorithm to check if an erasure pattern is correctable by an $(m,n,a=1,b)$-MR tensor code in time $m(m+n)^3.$
\end{theorem}

\subsection{Proof Overview}
\paragraph{Higher order MDS - MR tensor code equivalence.} The proof of the equivalence between $(m,n,a=1,b)$-MR tensor codes and $\MDS(m)$ codes follows from some linear algebra and inductive arguments.

\paragraph{Field size lower bound.} Our lower bound for $\MDS(m)$ is
inspired by lower bounds for Maximally Recoverable Local
Reconstruction Codes from \cite{gopi2020maximally} and works as
follows. We will first prove a lower bound when the code dimension is
small. 
{We use the probabilistic method to show that if the field size is too small, then there will be subspaces which intersect non-trivially, but which shouldn't generically. See the discussion before Lemma~\ref{lem:weak-MDS-lowerbound} for a high-level overview of the proof.}
We then prove {a} weak duality for $\MDS(m)$ codes, which implies the lower
bound when the codimension is small.

\paragraph{Efficient correctability checking when $a=1$.} 
The previous work of \cite{GHKSWY} showed that $E$ is a correctable pattern for a $(m,n,a=1,b)$-MR tensor code if and only if the erasure pattern {satisfies} a combinatorial condition called \emph{regularity}. The regularity condition upper bounds the size of the intersection of $E$ with any subrectangle of $[m] \times [n]$. We show that the inequalities in this regularity condition correspond to capacity constraints in a suitable max-flow problem. We call this property \emph{excess-{compatibility}}. We show that excess-compatibility is {equivalent} to regularity by applying a generalization of Hall's marriage theorem on the existence of matchings \cite{bokal2012generalization}. As a result, we show that checking regularity is equivalent to a polynomial-sized maximum flow problem.

{ The work of Shivakrishna,
  et.al.~\cite{Shivakrishna_MRproduct} also considers a notion similar
  to excess-compatibility. In particular, they show that regularity of
  a pattern implies certain matching conditions on the bipartite graph
  induced by the erasure pattern (i.e., Lemmas II.2 and II.4 of their
  paper), and thus their results can be viewed as an analogue of the
  ``regularity implies excess-compatibility'' half of the proof
  Theorem~\ref{thm:regular-iff-excess}. However, our work appears to
  be the first to show that excess-compatibility is equivalent to regularity and the first to give a polynomial time algorithm for checking regularity.
}

\subsection{Open Questions}
The biggest open question is to obtain constructions of $(n,n-b)$-$\MDS(m)$ codes (or equivalently $(m,n,a=1,b)$-MR tensor codes). Our lower bounds show that we need fields of size at least $\Omega_{m,b}\lp n^{\min\{b,m\}-1} \rp$ whereas the upper bounds (which are not explicit) are $q=O_{m,b}(n^{b(m-1)})$. Closing this gap and getting explicit constructions of $\MDS(m)$ codes over small fields is the main question we leave open. Concretely, our lower bound shows that $(n,n-3)$-$\MDS(3)$ codes require fields of size $q\ge \Omega(n^2)$. We conjecture that this is tight. By $\MDS(3)$ duality (Corollary~\ref{cor:mds3-duality}), this is equivalent to constructing $(n,3)$-$\MDS(3)$ codes.
\begin{conjecture} 
	There exist $(n,3)$-$\MDS(3)$ codes over fields of size $q=O(n^2).$
\end{conjecture}
We give evidence for this conjecture by giving constructions of codes over fields of size $q=O(n^2)$ which come very close to being $\MDS(3)$, see Appendix~\ref{app:constructions} for these constructions. In these constructions, we relax (\ref{eqn:MDSell}) by restricting which sets $A_1,A_2,A_3$ we consider.  For the first construction, we split $n$ into two halves, and require that each $A_i$ is a two-element set using one element from each set.\footnote{This may seem very restrictive, but by Lemma~\ref{lemma:ultimate-mds-equiv}, we assume that $|A_1|+|A_2|+|A_3| = 6$ and even each $|A_i| = 2$ (as long as the code is MDS).} For the second construction, we split $n$ into three parts and require that $A_i$ is a two-element subset of the $i$th part.

Another important open question is further characterizing correctable patterns for MR tensor codes.

\begin{question}
  Can we efficiently detect which erasure patterns of an $(m,n,a,b)$-MR tensor code are correctable when $a,b>1$?
\end{question}

Currently, there is no ``simple'' condition like that of regularity and no efficient deterministic algorithm which is known for testing correctability when $a,b>1$. We shall investigate this question more deeply in a future work.

{\paragraph{Subsequent Work.} A very recent work of
  Roth~\cite{roth2021higher} defined another notion of higher-order
  MDS codes. This type of higher-order MDS code is motivated by
  applications to list decoding. In a follow-up
  work~\cite{brakensiek2022generic}, we show that Roth's notion and
  our notion of higher-order MDS codes are essentially equivalent, up
  to taking the dual of the code. This surprising connection between
  MR tensor codes and list decodability has a number of applications,
  including resolving in the affirmative the long-standing open
  question of whether there exists Reed-Solomon codes achieving
  list-decoding capacity.}

\subsection*{Organization}

In Section~\ref{sec:prelim}, we formally define MR tensor codes. In Section~\ref{sec:higherMDS}, we develop the theory of $\MDS(m)$ codes, including studying its duality properties. In Section~\ref{sec:lower-bounds}, we shows the field size lower bounds for $\MDS(m)$ and MR-tensor codes. In Section~\ref{sec:regularity}, we show how to efficiently {test} whether an erasure pattern is correctable for $(m,n,a=1,b)$-MR tensor codes.

In Appendix~\ref{app:upper-bound}, we prove Theorem~\ref{thm:upper} on randomized constructions of MR-tensor codes. In Appendix~\ref{app:mds-extra}, we provide a number of proofs omitted from the main exposition. In Appendix~\ref{app:constructions}, we give constructions of codes which partially satisfy the $\MDS(3)$ property. 

\subsection*{Acknowledgements}
We would like to thank Sergey Yekhanin and Venkatesan Guruswami for
helpful discussions and encouraging us to work on this problem. We
also thank anonymous reviewers for helpful feedback on the paper.

\section{Preliminaries}\label{sec:prelim}
Let $\F$ be any field. Let $n > k \ge 1$ be integers. For any $A \subset [n]$ and a matrix $V$ with columns $v_1,v_2,\dots,v_n$, we use $V_A$ to denote the submatrix of $V$ formed by columns $\{v_i : i \in A\}$. 

A $(n,k)$-code $C$ is a $k$-dimensional subspace of $\F^n$. It can either be described using a generator matrix $G_{k\times n}$ such that $C=\set{G^Tx: x\in \F^k}$ or using a parity check matrix $H_{(n-k) \times n}$ such that $C=\set{y\in \F^n: Hy=0}$. Note that $G$ and $H$ have full row rank and $HG^T=0$. The rows of $G$ form a basis for $C$ and the rows of $H$ form a basis for the dual code $C^\perp$. Note that $G,H$ are not uniquely determined by $C$, but they are unique up to basis change.

\begin{proposition}
\label{prop:correctability}
Let $C$ be a $(n,k)$-code with generator matrix $G_{k\times n}$ and parity check matrix $H_{(n-k)\times n}$. Let $E\subset [n]$ be an erasure pattern and let $\barE=[n]\setminus E$. 
The following conditions are equivalent.
\begin{enumerate}
	\item $E$ is correctable i.e. given $x_\barE$ for some unknown $x\in C$, we can recover $x.$
	\item $G_\barE$ has rank $k$.
	\item $H_E$ has full column rank.
\end{enumerate}
\end{proposition}

Therefore a maximal correctable erasure pattern has size $n-k.$ And codes which correct all erasure patterns of size $n-k$ are called MDS codes.
\begin{definition}(MDS Code)
	\label{def:MDS}
	A $(n,k)$-code is called an MDS code if it can correct every erasure pattern of size $n-k.$
\end{definition}
Reed-Solomon codes are explicit MDS codes and they can be constructed for all $k,n$ over fields of size $O(n)$ which is tight. We now present several equivalent properties of MDS codes.
\begin{proposition}
\label{prop:MDS}
Let $C$ be a $(n,k)$-code with generator matrix $G_{k\times n}$ and parity check matrix $H_{(n-k)\times n}$.
The following conditions are equivalent.
\begin{enumerate}
	\item $C$ has distance $n-k+1.$
	\item Every erasure pattern of size at most $n-k$ is correctable.
	\item Every $k\times k$ minor of $G$ is non-zero.
	\item Every $(n-k)\times(n-k)$ minor of $H$ is non-zero.
\end{enumerate}
\end{proposition}

\section{Higher order MDS codes}\label{sec:higherMDS}
\subsection{Basic properties of higher-order MDS codes}

In this section, we will prove some properties of higher-{order} MDS codes that we will need. The proofs are given in Appendix~\ref{app:mds-extra}. The following proposition gives an equivalent definition of $\MDS(\ell)$ codes.

\begin{restatable}{lemma}{lemmaone}\label{lemma:ultimate-mds-equiv}
  Let $V \in \mathbb F^{k \times n}$ be an $(n,k)$-MDS code and let $\ell \ge 2.$ Let $W \in \mathbb R^{k \times n}$ be a generic real matrix. Then $V$ is $\MDS(\ell)$ if and only if for all $A_1, \hdots, A_{\ell} \subseteq [n]$ such that $|A_i| \le k$, $|A_1| + \cdots + |A_{\ell}| = (\ell-1)k$ and $A_1  \cap \cdots \cap A_k = \emptyset$, we have that
   \[
     V_{A_1} \cap V_{A_2} \cap \cdots \cap V_{A_\ell} = 0 \iff W_{A_1} \cap \cdots \cap W_{A_{\ell}} = 0,
   \]
\end{restatable}

The following proposition shows that $\MDS(\ell)$ property is
{preserved} under puncturing and shortening of codes. If $C$ is
any $(n,k)$ code, the punctured code at position $i$ is an $(n-1,k)$
code obtained given by projecting all the codewords of $C$ onto the
subset $[n]\setminus \{i\}$. The shortened code at position $i$ is an
$(n-1,k-1)$ code obtained by projecting only the codewords {$x
  \in C$ for which $x_i = 0$} onto $[n]\setminus\{i\}$.

\begin{restatable}{proposition}{propone}\label{prop:reduction}
   Let $C$ be an $(n,k)$-$\MDS(\ell)$ code.
   \begin{enumerate}
     \item {If $\ell \ge 3$, then} $C$ is also an $\MDS(\ell-1)$ code.
     \item {If $\ell \ge 2$, then the} code $C_0$ obtained by puncturing $C$ at any position is an $(n-1,k)$-$\MDS(\ell)$ code.
     \item {If $\ell \ge 2$, then the} code $C_1$ obtained by shortening $C$ at any position is an $(n-1,k-1)$-$\MDS(\ell)$ code.
   \end{enumerate}
\end{restatable}

\subsection{Equivalence between MR tensor codes with $a=1$ and $\MDS(\ell)$}

The following lemma shows that $\MDS(\ell)$ codes are intimately connected to MR Tensor codes with $a=1$ or $b=1.$
\begin{lemma}
	\label{lem:correctability_intersection_a1}
	 Let $C=C_\col \otimes C_\row$ be an $(m,n,a=1,b)$-Tensor Code where $C_\col$ is a parity check code and $C_\row$ is an $(n-b,n)$-MDS code.  
	 Let $E\subset [m]\times [n]$ be a maximal erasure pattern
         i.e. $|E|=mn-(m-1)(n-b)$ and suppose that each row has at
         least $b$ erasures. Let {$A_1, \hdots, A_m \subset [n]$
         such that $\cup_{i=1}^m \set{i}\times A_i = \barE.$} Then $E$ is correctable iff $\dim(V_{A_1}\cap V_{A_2}\cap \dots \cap V_{A_m})=0.$ 
\end{lemma}
\begin{proof}
	The conditions on $E$ translate to the following conditions on $A_1,A_2,\dots, A_m.$
	\begin{enumerate}
		\item $|A_i|\le n-b$
		\item $\sum_i |A_i|= (m-1)(n-b).$ 
	\end{enumerate}
	Since $C_\col$ is a simple parity check code, each column of $C$ sum to zero. Let $V_{\barb \times n}$ be a generator matrix for $C_\row$ where $\barb=n-b.$ The following statements are equivalent.
	\begin{enumerate}
		\item $E$ is not correctable.
		\item {By definition of correctability, there} exists a non-zero codeword of $C$
                  {whose support is a subset of} $E.$ 
		\item {From the code being a tensor,} there exist $r_1,r_2,\dots,r_m\in C_\row$, not all zero, such that 
			\begin{itemize}
				\item $\supp(r_i)\subset \barA_i$ for $i\in [m]$,
				\item $\sum_{i=1}^m r_i = 0.$
			\end{itemize}
			{Since} $r_i=y_i^T V$ for some $y_i\in \F^\barb$, we have the following equivalent statement.
		\item There exist $y_1,y_2,\dots,y_m \in \F^{\barb}$, not all zero, such that 
			\begin{itemize}
				\item $y_i^T V_{A_i} =0$ for $i\in [m]$,
				\item $\sum_{i=1}^m y_i =0.$
			\end{itemize}
			{Since} $y_i \in V_{A_i}^\perp$ for each $i\in [m],$ we have the following equivalent statement.
		\item There exists $y_i\in V_{A_i}^\perp$, not all zero, such that $\sum_{i=1}^m y_i =0$.\\ 
		{Since} $|A_i|\le n-b$, and $V$ is a generator matrix of an MDS code, we have $\dim(V_{A_i}^\perp)=(n-b)-\dim(V_{A_i})=(n-b)-|A_i|$. Therefore $$\sum_{i=1}^m \dim(V_{A_i}^\perp) = \sum_{i=1}^m (n-b-|A_i|)=m(n-b)-(m-1)(n-b)=n-b.$$ So we have the following equivalent statement.
		\item $V_{A_1}^\perp + V_{A_2}^\perp + \dots + V_{A_m}^\perp \ne \F^{\barb}.$
		\item {By taking the dual,} $V_{A_1}\cap V_{A_2}\cap \dots \cap V_{A_m} \ne 0.$
	\end{enumerate}
	This completes the proof.
\end{proof}

We can now prove Theorem~\ref{thm:MRtensorcodesa1_MDSm}.
\begin{corollary}[(Theorem~\ref{thm:MRtensorcodesa1_MDSm})]\label{cor:MDS-MR}
	Let $C=C_\col \otimes C_\row$ be an $(m,n,a=1,b)$ tensor code. Let $C_\col$ be the parity check code. Then $C$ is an MR tensor code iff $C_\row$ is $(n,n-b)-\MDS(m).$
\end{corollary}
\begin{proof}
	If $C_\row$ is $\MDS(m)$, then by Lemma~\ref{lem:correctability_intersection_a1}, $C$ is MR. The other direction follows from Lemma~\ref{lem:correctability_intersection_a1} and Lemma~\ref{lemma:ultimate-mds-equiv}. 
\end{proof}

In the style of Lemma~\ref{lem:correctability_intersection_a1}, we
prove a similar but more intricate lemma which captures the case $a
\ge 2$. Recall we have $a \le m$ and $b \le n$. {This result
  is used to prove the weak duality of higher-order MDS codes.}

\begin{restatable}{lemma}{lemmatwo}
\label{lem:correctability_intersection}
Let $C=C_\col \otimes C_\row$ be an $(m,n,a,b)$ tensor code and let
$\bar{a}=m-a, \bar{b}=n-b$. Let $E$ be a maximal erasure pattern of
size $|E|=mn-\bar{a}\bar{b}.$ and let $\bar{E}=\cup_{i\in
  [m]}\set{i}\times A_i=\cup_{j\in [n]}B_j \times \set{j}.$ If $U,V$
are generator matrices of $C_\row$ and $C_\col$ and $P,Q$ are their
respective parity check matrices{, then} correctability of $E$ is equivalent to each of the following conditions:
\begin{align}
  \sum_{i=1}^m U_i \otimes V_{A_i} &= \F^{\bar{a}} \otimes \F^{\bar{b}}\label{eq:2}\\
  \sum_{j=1}^n U_{B_j} \otimes V_j &= \F^{\bar{a}} \otimes \F^{\bar{b}}\label{eq:3}\\
  \sum_{i=1}^m P_i \otimes V_{A_i}^\perp &= \F^a \otimes \F^{\bar{b}}\label{eq:4}\\
  \sum_{j=1}^n U_{B_j}^\perp \otimes Q_j &= \F^{\bar{a}} \otimes \F^b\label{eq:5}.
\end{align}
\end{restatable}

Notice that if $a=1$ and $C_\col$ is a parity check code, then $P_i = 1$ for all $i$. So the expression becomes $\sum_{i=1}^m V_{A_i}^\perp = \F^{\bar{b}}$. Taking duals, we get $\bigcap_{i=1}^m V_{A_i}=0$ which is equivalent to Lemma~\ref{lem:correctability_intersection_a1}.

\subsection{Weak duality of $\MDS(\ell)$}

A natural conjecture is that for all $\ell \ge 2$, a code $C$ is $(n,k)$-$\MDS(\ell)$ iff the dual code $C^{\perp}$ is $(n,{n-k})$-$\MDS(\ell)$. This is true for $\ell=2$ because $\MDS(2)$ is equivalent to the usual MDS, and MDS codes satisfy duality. We will later show that duality also holds for $\ell=3$. But {surprisingly}, this fails for $\ell\ge 4.$ We exhibit a counterexample in Appendix~\ref{app:mds-extra}. 

  As duality of higher-order MDS codes is false in general, we instead prove a weaker form of duality--the dual of an $\MDS(\ell)$ code satisfies what we call ``cycle-$\MDS(\ell)$'' property which is a weaker form of $\MDS(\ell).$ This result will also imply that the dual of any $\MDS(3)$ code is indeed $\MDS(3)$.

\subsubsection{Cycle-MDS}

\begin{definition}[cycle-$\MDS(\ell)$]
Call a collection of subsets $S_1,S_2, \hdots, S_{\ell} \subset [n]$ a \emph{cycle family} if for all $j \in [n]$, the set $T_j := \{i : j \in S_i\}$ is an interval modulo $\ell$, i.e., $T_j = \{c_j,c_j+1, \hdots, d_j\} \mod \ell$ for some $c_j$ and $d_j$. If you visualize $S_1,S_2,\dots,S_\ell$ as subsets of rows of an $\ell \times n$ matrix, then $T_1,T_2,\dots,T_n$ are the subsets of columns corresponding to $S_1,S_2,\dots,S_\ell.$

  Say that an $k \times n$ matrix $V$ is $(n,k)$-cycle-$\MDS(\ell)$ if for any cycle family $S_1, \hdots, S_{\ell} \subset [n]$ such that $|S_i| \le k$ for all $i$, $|S_1| + \cdots + |S_{\ell}| \le (\ell-1)k$, and $S_1 \cap \cdots \cap S_{\ell} = \emptyset$, we have that $V_{A_1} \cap \cdots \cap V_{A_\ell} = 0$ if and only if it generically holds.
\end{definition}

\begin{lemma}\label{lem:weak-duality}
Let $C$ be an $(n,k)$-$\MDS(m)$ code. Then the dual code $C^\perp$ is {an} $(n,n-k)$-cycle-$\MDS(m)$ code.
\end{lemma}

\begin{proof}
  In the proof, we crucially use the fact that the generator matrix of
  any $(m,m-1)$-MDS code is equivalent to a generic $(m,m-1)$ matrix
  {up to} symmetries.\footnote{If $U$ is the generator matrix
    of an $(m,m-1)$-MDS code and $X_{(m-1)\times m}$ is any generic
    matrix, we can scale the columns of $X$ and change basis (left
    multiply with an invertible $(m-1)\times (m-1)$ matrix) to make it
    equal to $U.$} In particular any $(m,m-1)$-MDS code is also
  $\MDS(\ell)$ for all $\ell\ge 2.$ This is because there is a unique
  $(m,1)$-MDS code {up to} symmetries which is the parity check
  code (whose generator matrix has a single row of all ones).

  Let $V_{k\times n}$ be the generator matrix of $C$ and let $Q_{(n-k)\times n}$ be its parity check matrix. Note that $Q$ is also the generator matrix for the dual code $C^\perp.$
  Let $a = 1$, and $b = n-k$. Let $U_{(m-1)\times m}$ be a generator matrix of an $(m,m-1)$-MDS code (and thus also $MDS(m)$). %

  Let $S_1,S_2,\dots, S_{m} \subseteq [n]$ be a cycle family, each of size at most $b=n-k$ and of total size $(m-1)b$ and $\bigcap_{i=1}^m S_i=\emptyset$.
  To show that $Q$ is $(n,b)$-cycle-$\MDS(m)$, it suffices to show that $Q_{S_1} \cap \cdots \cap Q_{S_m} = 0$ whenever it holds generically.

  For all $j \in [n]$, let $T_j := \{i \in [m] : j \in S_i\}$. Recall that since $S_1,S_2,\dots,S_m$ is a cycle-MDS family, we have $T_j = \{c_j, \hdots, d_j\} \mod m$ for some $c_j, d_j$ or $T_j  = \emptyset$. Note that $|T_j|\le m-1$ since $\bigcap_{i=1}^m S_i =\emptyset.$

  Construct $B_1,B_2, \hdots, B_{n}\subset [m]$ as follows (again indices are {considered} modulo $m$):
  \[
    B_i = \begin{cases}
      [m] \setminus \{c_j, \hdots, d_j+1\} & \text{ if $T_j \neq \emptyset$}\\
      [m-1] & \text{ otherwise}
    \end{cases}
  \]
  Note that each $|B_i| \le m-1$ and
  \[
    \sum_{i=1}^{n} |B_i| = (m-1)n - \sum_{j=1}^{m} |S_j| = (m-1)n - (m-1)b = (m-1)k.
  \]

  For all $i \in [m]$, let $W_i := U_{[m] \setminus \{i,i+1\}}^{\perp},$ where indices are taken modulo $m$.
  \begin{claim}\label{claim:wMDS}
    For any MDS $U$, the family $W_1, \hdots, W_{m}$ is MDS.
  \end{claim}
  \begin{proof}
    By symmetry, it suffices to show that $W_1, \hdots, W_{m-1}$ are linearly independent. Consider the following matrix product.
    \begin{align*}
      \begin{bmatrix}
        U_1^T\\
        U_2^T\\
        \vdots\\
        U_{m-1}^T
      \end{bmatrix}\cdot  
      \begin{bmatrix}
        W_1 & W_2 &\cdots &W_{m-1}
      \end{bmatrix}=
      \begin{bmatrix}
        * &  &  & & \\
        * & * &  &  & \\
         & * & * &  & \\
         &  &\ddots &\ddots&\\
         &  &  & * &*\\
      \end{bmatrix}
    \end{align*}
    where $*$ corresponds to a non-zero entry and all the unmarked entries are $0$. Here we used that fact that $\inpro{W_i}{U_i}\ne 0$ and $\inpro{W_i}{U_{i+1}}\ne 0$, but $\inpro{W_i}{U_j}=0$ for all $j\notin \{i,i+1\}$, which follows from the MDS property of $U.$ Since the RHS matrix is clearly full rank, the matrices on the LHS product are both full rank.
    Therefore, $W_1, \hdots, W_{m}$ is indeed MDS.
  \end{proof}

  \begin{claim}
    \label{claim:switchU_W}
    $U_{B_i}^{\perp} = W_{T_j}=\operatorname{span}\{W_{c_j}, \hdots, W_{d_j}\}.$
  \end{claim}
 \begin{proof}
  Clearly $W_{c_j},\dots,W_{d_j}\in U_{B_i}^\perp$. Since they are part of an MDS code of dimension $m-1$ and $|T_j|\le m-1$, $W_{c_j},\dots,W_{d_j}$ are linearly independent. Finally by the MDS property of $U$ and since $|B_i|\le m-1$, $$\dim(U_{B_i}^\perp)=(m-1)-|B_i|=(m-1)-(m-(|T_i|+1))=|T_i|.$$ Therefore we get an equality by dimension counting.
 \end{proof}
By Lemma~\ref{lem:correctability_intersection}, 
  \begin{align*}
    &\bigcap_{i=1}^m Q_{S_i}=0\\
    \iff & \sum_{i=1}^m W_i \otimes Q_{S_m}=\F^{m-1} \otimes \F^b\\
    \iff & \sum_{j=1}^n W_{T_j} \otimes Q_i=\F^{m-1} \otimes \F^b\\
    \iff & \sum_{j=1}^n U_{B_j}^\perp \otimes Q_i=\F^{m-1} \otimes \F^b\tag{By Claim~\ref{claim:switchU_W}}\\
    \iff & \sum_{j=1}^n U_{B_j} \otimes V_i=\F^{m-1} \otimes \F^{n-b}\\
    \iff & \sum_{i=1}^m U_i \otimes V_{A_i} = \F^{m-1}\otimes \F^{n-b} \tag{$A_i=\{j: i\in B_j\}$}\\
    \iff &  \bigcap_{i=1}^m V_{A_i}=0.
  \end{align*}
Since $V$ is $\MDS(m)$, it is enough to show that $\bigcap_{i=1}^m \tQ_{S_i}=0$ for a generic $\tQ_{b\times n}$ implies $\bigcap_{i=1}^m \tV_{A_i}=0$ for a generic $\tV_{(n-b)\times n}$. This follows by Lemma~\ref{lem:correctability_intersection} by the same chain of {equivalences} as above and noting that the dual of a generic code is also generic, i.e., the parity check matrix corresponding to a generic generator matrix is also generic.

\end{proof}

A special case of a cycle family is when the sets $S_1, \hdots,
S_{\ell}$ are all disjoint.

\begin{definition}[weak-$\MDS(\ell)$]
  Say that an $k \times n$ matrix $V$ is $(n,k)$-weak-$\MDS(\ell)$ if for any $\ell$ \emph{disjoint} subsets $S_1, \hdots, S_{\ell} \subset [n]$ such that $|S_i| \le k$ for all $i$ and $|S_1| + \cdots + |S_{\ell}| \le (\ell-1)k$, we have that $V_{S_1} \cap \cdots \cap V_{S_\ell} = 0$.
\end{definition}

{This notion is the minimal assumption needed of the structure
  of the code for our first field size lower bound (Lemma~\ref{lem:weak-MDS-lowerbound}) to hold.}

{
\begin{proposition}
\label{prop:cycleMDS_weakMDS}
Suppose $C$ is a $(n,k)$-cycle-$\MDS(\ell)$ code, then $C$ is also a $(n,k)$-weak-$\MDS(\ell)$ code.
\end{proposition}
\begin{proof}
Let $V$ be a generator matrix for $C$. Let $S_1,S_2,\dots,S_\ell \subset [n]$ be a mutually disjoint family of subsets such that $|S_1|+|S_2|+\dots+|S_\ell|\le (\ell-1)k.$ To prove that $C$ is weak-$\MDS(\ell)$, it is enough to show that $V_{S_1}\cap \dots \cap V_{S_\ell}=0.$

 Clearly $S_1,S_2,\dots,S_\ell$ is also a \emph{cycle family}; if we imagine $S_1,\dots,S_\ell$ as subsets of rows of an $\ell\times n$ matrix, then the corresponding subsets $T_j$ of columns are just singletons. Since $S_1,S_2,\dots,S_\ell$ are disjoint and $|S_1|+|S_2|+\dots+|S_\ell|\le (\ell-1)k$, for a generic $k\times n$ matrix $W$, we should have $W_{S_1}\cap \dots \cap W_{S_\ell}=0$. This is because $W_{S_1}^\perp,\dots,W_{S_\ell}^\perp$ are generic subspaces (here we are using mutual disjointness of $S_i$) such that $$\sum_{i=1}^\ell \dim(W_{S_i}^\perp) =\sum_{i=1}^\ell (k-|S_i|)= k\ell - \sum_{i=1}^n |S_i|\ge k$$ and so $\sum_{i=1}^\ell W_{S_i}^\perp = \F^k.$ Therefore, since $C$ is cycle-$\MDS(\ell)$, we should have $V_{S_1}\cap \dots \cap V_{S_\ell}=0.$
\end{proof}  

Thus, Lemma~\ref{lem:weak-duality} and Proposition~\ref{prop:cycleMDS_weakMDS} together imply that the dual of any $(n,k)$-$\MDS(m)$ code is a $(n,n-k)$-weak-$\MDS(m)$ code. This observation is useful in proving our field size lower bounds.
}
\subsubsection{$\MDS(3)$ duality}

It turns out that the cycle-$\MDS(m)$ condition is equivalent to the cycle-$\MDS(m)$ condition when $m=3$.

\begin{proposition}
  Let $V$ be a $(n,k)$-cycle-$\MDS(3)$ code. Then, $V$ is an $(n,k)$-cycle-$\MDS(3)$ code.
\end{proposition}

\begin{proof}
  By Lemma~\ref{lemma:ultimate-mds-equiv}, to prove that $V$ is
  $(n,k)-\MDS(3)$ it suffices to check that for all sets $A_1, A_2,
  A_3$ with $|A_i| \le k$, $|A_1|+|A_2|+|A_3| = 2k$, and $A_1 \cap A_2
  \cap A_3 = \emptyset$, we have that $V_{A_1} \cap V_{A_2} \cap
  V_{A_3} = 0$ if and only if it should happen generically. Since
  every subset of $\{1,2,3\}$ is either an {empty set} or an ``interval'' modulo $3$, we have that $A_1, A_2, A_3$ is a cycle family. Thus, the required condition holds because $V$ is $(n,k)$-cycle-$\MDS(3)$.
\end{proof}

\begin{corollary}\label{cor:mds3-duality}
The dual of an $(n,k)$-$\MDS(3)$ code is a $(n,n-k)$-$\MDS(3)$ code.
\end{corollary}

As mentioned previously, the dual of an $\MDS(4)$ code is not necessarily an $\MDS(4)$ code (but it must be an $\MDS(3)$ code).

\section{Lower bounds on field size for $\MDS(\ell)$}\label{sec:lower-bounds}
{Our lower bound for $\MDS(m)$ is
inspired by lower bounds for Maximally Recoverable Local
Reconstruction Codes from \cite{gopi2020maximally} and works as
follows. We will actually prove the lower bound for weak-$\MDS(m)$ codes which implies a lower bound for $\MDS(m)$ codes.
By a reduction, we can assume that the dimension of the code is equal to $m$. Suppose $V$ be the generator matrix of an $(n,m)$-weak-$\MDS(m)$ code where $m$ is a constant and suppose $m$ divides $n$ for simplicity. We will partition $V_1,V_2,\dots,V_n$ in to $m$ parts of size $n/m$ each, say the partition is given by $[n]=\cP_1 \sqcup \cP_2 \sqcup \dots \sqcup \cP_m$. Now consider arbitrary subsets $A_i\subset \cP_i$ of size $|A_i|=m-1$ for $i\in [m].$ For $V$ to be $\MDS(m)$, it is necessary that $V_{A_1}\cap V_{A_2} \cap \dots \cap V_{A_m}=0$. To see this, note that $x\in V_{A_i}$ imposes $1$ generic linear equation on $x$. And so $x\in V_{A_i}$ for all $i\in [m]$ imposes $m$ equations which should be linearly independent if $V$ behaves generically. Fix any $i\in [m]$ and fix a subset $A_i\subset \cP_i.$ A random point $X$ of $\F_q^m$ lies in $V_{A_i}$ for some $A_i \subset \cP_i$ with high probability. This is because for a fixed subset $A_i\subset \cP_i$, $X$ lies in $V_{A_i}$ with probability $\frac{1}{q}.$ Since there are $\binom{n/m}{m-1}$ subsets $A_i\subset \cP_i$, the expected number of $A_i\subset \cP_i$ such that $X\in V_{A_i}$ is $\binom{n/m}{m-1} \cdot \frac{1}{q}$ which is $\gg 1$ if $q\ll_{m} n^{m-1}$. By using pairwise independence and carefully calculating second moments, we can conclude that if $q\ll_{m} n^{m-1}$, with high probability, there exists some $A_i \subset \cP_i$ such that $X\in V_{A_i}.$ By union bound over $i\in [m]$,  $X\in V_{A_1}\cap V_{A_2} \cap \dots \cap V_{A_m}$ for some $A_i \subset \cP_i$ and $X$ is non-zero with high probability which violates the $\MDS(m)$ property. Therefore $q\gg_{m} n^{m-1}.$ We will now make this argument formal.}

\begin{lemma}\label{lem:weak-MDS-lowerbound}
Assume $n \gg k \ge 3$. Let $V$ be an $(n,k)$-weak-$\MDS(k)$ code over field $\F_q$. Then $q \ge \Omega_{k}(n^{k-1}).$
\end{lemma}
\begin{proof}
  Let $s = \lfloor n/k \rfloor$. {For} all $i \in [k]$, let $I_i = \{(i-1)s + 1, \hdots, is\}$. Let $k' = k-1$. Let $\mathcal S_i$ be all subsets of $I_i$ of size $k'$.

  Since $V$ is $(n,k)$-weak-$\MDS(k)$ we have that 
\begin{align}
\text{for all $A_1 \in \mathcal S_1, \hdots, A_{k} \in \mathcal S_{k}$,\ \ \ \ \ \ \ \ }\dim(V_{A_1} \cap \cdots \cap V_{A_{k}}) = 0.\label{eq:MDS-partition}
\end{align}

We seek to show that if $q\ll_k n^{k-1}$, then the above condition is violated. Consider the following random process. Sample $x \in \F^k$ uniformly at random, and for all $i \in [k]$, let $X_i$ be the number of $A_i \in \mathcal S_i$ such that $x \in V_{A_i}$. If we can show that with nonzero {probability} all the $X_i$'s are simultaneously nonzero and $x$ is nonzero, then we know that (\ref{eq:MDS-partition}) is violated.

Observe that for all $i \in [k]$ by linearity of expectation
\begin{align}
  \E[X_i] = \sum_{A_i \in \mathcal S_i} \Pr_{x \sim \F^k}[x \in V_{A_i}] = \binom{s}{k'} \frac{1}{q^{k-k'}} = \binom{s}{k-1} \frac{1}{q}.\label{eq:EX}
\end{align}
Note that $\E[X_i]\gg 1$ if $q\ll_k n^{k-1}$. To conclude that $\Pr[X_i>0]\approx 1$, we will show that second moment $\E[X_i^2]\approx \E[X_i]^2$ and use $\Pr[X_i>0]\ge \E[X_i]^2/\E[X_i^2]$. We will use the fact that for $A_i,A_i'\in \mathcal S_i$,
\begin{align*}
\dim(V_{A_i}\cap V_{A_i'})&=\dim(V_{A_i})+\dim(V_{A_i'})-\dim(V_{A_i\cup {A'_i}})\\
&=2k'-\min\{k,2k'-|A_i\cap A_i'|\}\\
&=\max\{2k'-k,|A_i\cap A_i'|\}\\
&=\max\{k-2,|A_i\cap A_i'|\}
\end{align*}

\begin{align*}
  \E[X^2_i] &= \sum_{A_i, A'_i \in \mathcal S_i} \Pr_{x \sim \F^k}[x \in V_{A_i} \cap V_{A'_i}]\\
            &= \sum_{j=0}^{k'}\sum_{\substack{A_i,A'_i \in \mathcal S_i\\|A_i \cap A'_i| = j}}\frac{q^{\max(j, k-2)}}{q^k}\\
            &= \sum_{j=0}^{k'} \binom{s}{k'}\binom{k'}{j}\binom{s-k'}{k'-j}\frac{q^{\max(j, k-2)} }{q^k}\\
            &= \E[X_i]^2 \sum_{j=0}^{k'} \frac{\binom{s}{k'}\binom{k'}{j}\binom{s-k'}{k'-j}}{\binom{s}{k'}^2}\cdot\frac{q^{\max(j, k-2)}\cdot q^{2}}{q^k}\\
            &= \E[X_i]^2 \sum_{j=0}^{k'} \frac{\binom{k'}{j}\binom{s-k'}{k'-j}}{\binom{s}{k'}}\cdot q^{\max(j-(k-2), 0)}\\
            &\le \E[X_i]^2 \lp 1+ O_{k}(1)\sum_{j=1}^{k-1} \frac{q^{\max(j-(k-2),0)}}{n^j}\rp\\
            & =  \E[X_i]^2 \lp 1+ O_k(1)\lp\frac{1}{n}+\frac{1}{n^2}+\dots+\frac{1}{n^{k-2}}+\frac{q}{n^{k-1}} \rp\rp\\
            & \le\E[X_i]^2\lp 1+\frac{1}{k}\rp. \tag{If $q\ll_k n^{k-1}$ and $n\gg_k 1$}
\end{align*}

Thus, for each $i \in [k]$,
\[
  \Pr[X_i > 0] \ge \frac{\E [X_i]^2}{\E [X_i^2]} \ge \frac{1}{1+1/k} \ge \frac{k}{k+1}.
\]
Therefore, by the union bound, all the $X_i$'s are at least $1$ simultaneously and the sampled $x \in \F^k$ is nonzero with probability at least $1-k\cdot \frac{1}{k+1}-\frac{1}{q^k}>0$, a contradiction.
\end{proof}

Combining Lemma~\ref{lem:weak-MDS-lowerbound}, Proposition~\ref{prop:reduction}, and Lemma~\ref{lem:weak-duality} allows us to prove Theorem~\ref{thm:main-lower-bound} and Corollary~\ref{cor:main-lower-bound}.

\begin{corollary}[Theorem~\ref{thm:main-lower-bound}]\label{cor:lower-bound}
  Let $C$ be an $(n,k)$-$\MDS(\ell)$ code over $\F_q$. Then $$q\ge \Omega_{\ell}\lp n^{\min\{k,n-k,\ell\}-1}\rp.$$
\end{corollary}
\begin{proof}
\textbf{Case 1: $k>n/2$.}
Let $b = \min(n-k, \ell)$. By Proposition~\ref{prop:reduction}, we can puncture the code $C$ at $n-k-b$ locations to get a code $C_0$ which is $(k+b,k)$-$\MDS(b)$ code. By Lemma~\ref{lem:weak-duality}, $C_0^{\perp}$ is $(k+b, b)$-cycle-$\MDS(b)$. By Lemma~\ref{lem:weak-MDS-lowerbound}, we have that the field size of $C_0^{\perp}$ (and thus $C$) is $$q\ge \Omega_{b}(k^{b-1})=\Omega_{\ell}\lp n^{\min\{\ell,n-k\}-1}\rp.$$\\

\textbf{Case 2: $k\le n/2$.}
Let $b = \min(k, \ell)$. By Proposition~\ref{prop:reduction}, we can shorten the code $C$ at $k-b$ locations to get a code $C_0$ which is $(n-k+b,b)$-$\MDS(b)$ code (also trivially cycle-$\MDS(b)$). By Lemma~\ref{lem:weak-MDS-lowerbound}, we have that the field size of $C_0^{\perp}$ (and thus $C$) is $$q\ge \Omega_{b}(n^{b-1})=\Omega_{\ell}\lp n^{\min\{\ell,k\}-1}\rp.$$
\end{proof}

\begin{corollary}[Corollary~\ref{cor:main-lower-bound}]
  Let $C=C_{\col} \otimes C_{\row}$ be an $(m,n,a,b)$-MR tensor code. The minimum field size of $C$ is at least $\Omega_{m}(n^{\min\{b-1,n-b-1, m-a\}}).$
\end{corollary}

\begin{proof}
 Like in Proposition~\ref{prop:reduction}, let $C'_{\row}$ be an $(m-a+1, m-a)$-MDS code formed by puncturing $C_{\row}$ at any $a-1$ positions. Observe that $C_{\col} \otimes C'_{\row}$ must be an $(m-a+1,n,1,b)$-tensor {code}.
 This is because the generator matrix of $C_{\col}\otimes C'_{row}$ is $U \otimes V_{[n]\setminus A}$ where $A$ is the set of punctured positions which is a part of $U\otimes V$, the generator matrix for $C_{\col}\otimes C_{\row}$. Moreover, there is a canonical injective map from correctable erasure patterns of $C_{\col}\otimes C'_{\row}$ and those of $C_{\col}\otimes C_{\row}$.
 By Corollary~\ref{cor:MDS-MR}, $C_{\col}$ is an $(n, n-b)$-$\MDS(m-a+1)$ code. Thus, by Theorem~\ref{thm:main-lower-bound} the field size of $C_{\col}$ is at least $\Omega_{m}(n^{\min\{b,n-b,m-a+1\}-1})$, as desired.
\end{proof}

\section{Efficient regularity testing for MR Tensor Codes when $a=1$}\label{sec:regularity}
In this section, we will assume that $a=1$. Therefore WLOG, we can assume that $C_\col$ is a parity check code. In this case, we have a neat characterization of the generically correctable patterns in terms of \emph{regularity}.

\subsection{Characterizing correctable patterns: Regularity}

\begin{definition}[Regular pattern{, \cite{GHKSWY}}]
	An erasure pattern $E\subset [m]\times [n]$ is regular for an $(m,n,a,b)$-Tensor Code if for every $S\subset [m]$ of size at least $a$ and $T\subset [n]$ of size at least $b$, we have 
	\begin{equation}
	\label{eqn:regularity_E}
	|E \cap (S \times T)| \le sb + ta - ab
	\end{equation} 
	where $s=|S|$ and $t=|T|.$
\end{definition}
\begin{remark}
	We can rewrite the regularity condition (\ref{eqn:regularity_E}) as:	
	\begin{equation}
	\label{eqn:regularity_barE}
	|\barE \cap (S \times T)| \ge (s-a)(t-b)
	\end{equation} 
\end{remark}

We will first prove that correctability implies regularity for any value of $a,b$. We will need the following key lemma.
\begin{lemma}
	\label{lem:intersection_tensor_product}
	If $U_1,U_2$ are subspaces of $U$ and $V_1,V_2$ are subspaces of $V$, then $$\dim(U_1\otimes V_1 + U_2 \otimes V_2)=\dim(U_1)\cdot\dim(V_1)+\dim(U_2)\cdot\dim(V_2)-\dim(U_1 \cap U_2)\cdot \dim(V_1\cap V_2).$$
\end{lemma}
\begin{proof}
	It is enough to show that $(U_1\otimes V_1)\cap (U_2 \otimes V_2)=(U_1 \cap U_2)\otimes (V_1 \cap V_2).$ The rest follows from the fact that for any two spaces $A,B$, $\dim(A+B)=\dim(A)+\dim(B)-\dim(A\cap B)$ and $\dim(A\otimes B)=\dim(A)\cdot \dim(B)$. %
		We will show that the duals of both sides are
                equal. {We use the fact that for any subspaces $A$ and
                $B$, we have that $A^{\perp} + B^{\perp} = (A \cap B)^{\perp}$.}
		\begin{align*}
			\left((U_1\otimes V_1)\cap (U_2 \otimes V_2)\right)^\perp &= (U_1\otimes V_1)^\perp +(U_2 \otimes V_2)^\perp\\
			&=  ( U_1^\perp \otimes V + U \otimes V_1^\perp ) + (U_2^\perp \otimes V + U \otimes V_2^\perp)\\
			&=  (U_1^\perp \otimes V + U_2^\perp \otimes V) + (U \otimes V_1^\perp + U \otimes V_2^\perp)\\
			&= (U_1^\perp + U_2^\perp) \otimes V + U \otimes (V_1^\perp + V_2^\perp)\\
			&= (U_1 \cap U_2)^\perp \otimes V + U \otimes (V_1\cap V_2)^\perp\\
			&= ((U_1\cap U_2)\otimes (V_1\cap V_2))^\perp\qedhere.
		\end{align*}
\end{proof}

\begin{theorem}[Correctability $\Rightarrow$ Regularity~\cite{GHKSWY}]\label{thm:correct-to-regular}
	If an erasure pattern $E$ is correctable by an $(m,n,a,b)$-tensor code, then it is regular.
\end{theorem}
\begin{proof}
By Proposition~\ref{prop:correctability}, $E$ is correctable iff $\dim((U\otimes V)_\barE)=(m-a)(n-b).$ Let $S\subset[m]$ be of size at least $a$ and $T\subset[n]$ be of size at least $b$. We can upper bound $\dim((U\otimes V)_\barE)$ as:
\begin{align*}
	\dim((U\otimes V)_\barE) &\le \dim((U\otimes V)_{\barE \cap (S\times T)}) + \dim((U\otimes V)_{\barE\setminus (S\times T)})\\
	&\le |\barE \cap (S\times T)| + \dim((U\otimes V)_{([m]\times [n])\setminus (S\times T)}).
\end{align*}
We now use Lemma~\ref{lem:intersection_tensor_product} to calculate $\dim((U\otimes V)_{([m]\times [n])\setminus (S\times T)})$.
\begin{align*}
\dim((U\otimes V)_{([m]\times [n])\setminus (S\times T)})&= \dim((U\otimes V)_{\barS\times [n]}+(U\otimes V)_{[m]\times \barT})\\
&=\dim(U_\barS \otimes V_{[n]} + U_{[m]}\otimes V_\barT)\\
&=\dim(U_\barS)\cdot \dim(V_{[n]}) + \dim(U_{[m]})\cdot \dim(V_\barT)-\dim(U_\barS)\cdot \dim(V_\barT)\\
&=(m-s)(n-b)+(m-a)(n-t)-(m-s)(n-t).
\end{align*}
Combining the above, we get $|\barE \cap (S\times T)|\ge (s-a)(t-b).$
\end{proof}

The following theorem shows that every regular pattern is correctable if $a=1$ or $b=1.$ 
\begin{theorem}[\cite{GHKSWY}] \label{Gop:Reg}
	\label{thm:regularity_correctability_a=1}
	{An} erasure pattern $E\subset [m]\times [n]$ is generically correctable for an $(m,n,a=1,b)$-tensor code iff $E$ is regular.
\end{theorem}

\begin{remark}
	\cite{GHKSWY} conjecture that regularity is equivalent to correctability for arbitrary $a,b$. This conjecture was recently disproved by \cite{holzbaur2021correctable}. For their counterexample, they consider $m = n = 5, a = b = 2$ and let the complement of the erasure pattern be
\[\barE  = \{(1, 1), (2, 2), (2, 3), (3, 2), (3, 3), (4, 4), (4, 5), (5, 4), (5, 5)\}.\]
$E$ is a regular pattern which is not correctable. Our analysis gives
a short proof of noncorrectability.  By
Proposition~\ref{prop:correctability}, $E$ is correctable iff
$\dim((U\otimes V)_\barE)=(m-a)(n-b)=9.$ {Because
  $\dim(U_{\set{2,3}} \cap U_{\set{4,5}}) \ge \dim(U_{\set{2,3}}) +
  \dim(U_{\set{4,5}}) - 3 = 1$ and likewise
  $\dim(V_{\set{2,3}} \cap V_{\set{4,5}}) \ge 1$, by
    Lemma~\ref{lem:intersection_tensor_product}, we have that}
\[{\dim(U_{\set{2,3}}\otimes
V_{\set{2,3}} + U_{\set{4,5}} \otimes V_{\set{4,5}}) \le 4 + 4 - 1 =
7,}\]
{Thus, $\dim((U\otimes V)_{\barE}) \le 8$, proving that $E$ is regular
but not correctable.}
\end{remark}

\subsection{Efficiently Checking Regularity}

Let $E \subset [m]\times [n]$ be any pattern for a $(m,n,a,b)$ tensor code. For all $i \in [m]$, let $\deg_E(i)$ be the number of $j \in [n]$ such that $(i,j) \in E$. Likewise, for all $j \in [n]$, define $\deg_E(j)$ to be the number of $i \in [m]$ such that $(i, j) \in E$. Let $f : [n] \to \mathbb Z_{\ge 0}$ be \emph{supply constraints} and $g : [m] \to \mathbb Z_{\ge 0}$ be \emph{demand constraints}. We define a $(g,f)_E$-\emph{quasi-matching} (c.f., Definition 2 of \cite{bokal2012generalization}) to be a subset of the edges $E' \subset E$ such that all supply/demand constraints are met:
\begin{enumerate}
\item For all $i \in [m]$, $\deg_{E'}(i) \ge g(i)$.
\item For all $j \in [n]$, $\deg_{E'}(j) \le f(j)$. 
\end{enumerate}
We shall use the following Hall-like condition for testing if a quasi-matching exists:

\begin{lemma}[Theorem 20 of \cite{bokal2012generalization}]\label{lem:hall}
  There does not exist a $(g,f)_E$-quasi-matching if and only if there exists a \emph{Hall-blocker}, that is a $U \subseteq [m]$ with
  \[
    \sum_{i \in U} g(i) > \sum_{j \in [n]} \min(f(j), \deg_{E \cap (U \times [n])}(j))
  \]
\end{lemma}

\begin{definition}
  Fix $V \subset [n]$ and let $E_V := E \cap ([m] \times V)$. For all $i \in [m]$ define the \emph{excess} $e(i) = \max (\deg_{E}(i) - b, 0)$. For all $j \in [n]$ define $a(j) := a$. Define a $V$-\emph{excess flow} $E' \subseteq E$ to be an $(e, a)_{E_V}$-quasi-matching.

  We say that a $E$ is \emph{excess-compatible} if for all $V \subset [n]$ of size $\bar{b} = n - b$, there exists a $V$-excess flow $E'_V \subset E_V$.
\end{definition}

{See also Construction II.3 in \cite{Shivakrishna_MRproduct} for a similar notion in the literature.}

\begin{theorem}\label{thm:regular-iff-excess}
$E$ is regular if and only if $E$ is excess-compatible.
\end{theorem}

\begin{proof} We proceed by showing both directions.

\paragraph{Regularity implies excess-compatibility.} First, assume $E$
is regular but not excess-compatible. Thus, there exists $V \subset
[n]$ of size $\bar{b}$ such that $E_V := E \cap ([m] \times V)$ lacks
a $V$-excess flow. {For any $U \subseteq [m]$, let $N(U)$ be
  the set of $j \in [n]$ for which there is $i \in U$ for which $(i,j)
\in E$.}
By Lemma~\ref{lem:hall}, there exists a Hall-blocker $U_1 \subset [m]$
and neighborhood $V_1 := N(U_1)\cap V$ for which demand exceeds supply. That is, if we let $E_1 = E \cap (U_1 \times V_1)$ then
\begin{align}
    \sum_{i \in U_1} e(i) > \sum_{j \in V_1} \min(a, \deg_{E_1}(j)).\label{eq:block}
\end{align}
We may assume without loss of generality that $e(i) \ge 1$ for all $i
\in U_1$--deleting any exceptions would keep the LHS the same and
perhaps decrease the RHS. In particular, $e(i) + b = \deg_{E}(i)$ for
all $i \in U_1$. {We seek to show a contradiction by proving
  that $n - |V| \ge b$.}

Let $V_2 \subset V_1$ be the vertices $j \in V_1$ for which
$\deg_{E_1}(j) < a$. Let $V_3 = V_1 \cup ([n] \setminus
V)$. {Note that $V_3 \setminus V_1 = [n] \setminus V.$} Thus, since $E$ is regular,
\begin{align*}
  |U_1| b + |V_3 \setminus V_2| a - ab &\ge |E \cap (U_1 \times (V_3 \setminus V_2))|\\
                                       &= |E \cap (U_1 \times V_3)| - |E \cap (U_1 \times V_2)|\\
                                       &= \sum_{i \in U_1} \deg_{E} (i) - \sum_{j \in V_2} \deg_{E \cap (U_1 \times V_2)} (j)\tag{no edges from $U_1$ to $V \setminus V_1 = [n] \setminus V_3$}\\%
                                       &= \sum_{i \in U_1} \deg_{E} (i) - \sum_{j \in V_2} \deg_{E_1} (j) \tag{extra vertices on right side do not change right-degrees}\\
                                       &= \sum_{i \in U_1} (b+e(i)) - \sum_{j \in V_2} \deg_{E_1} (j) \\
                                       &> |U_1| b + \sum_{j \in V_1} \min(a, \deg_{E_1}(j)) - \sum_{j \in V_2} \deg_{E_1} (j) \tag{by (\ref{eq:block})}\\
                                       &= |U_1| b + |V_1 \setminus V_2| a\tag{by definition of $V_2$}
\end{align*}
 	Therefore {since $V_2 \subset V_1$, $a(|V_3 \setminus
        V_1| - b) = a(n - |V| - b) > 0$}, which contradicts that $n - |V| = b${.}

\paragraph{Excess-compatibility implies regularity.} Second, assume $E$ is excess-compatible. We show by induction on the size of $U \subseteq [m]$ that for all $S \subseteq [n]$, with $|S| \ge b$, we have {that}
\[|E \cap (U \times S)| \le |U|b + |S|a - ab.\]
The base case of $U = \emptyset$ is trivial as the LHS equals $0$ for all $S$. For nontrivial $U$ we may assume for all $i \in U$ that $\deg_{{E}}(i) \ge b+1$. Otherwise, let $U' \subset U$ be the set of all $i$ with $\deg_{{E}}(i) \ge b+1$ and note that for all $S \subseteq [n]$ of size at least $b$, we have by the induction hypothesis that 
\begin{align*}
  |E \cap (U \times S){|} &\le |E \cap (U' \times S)| + b|U \setminus U'|\\
                       &\le |U'|b + |S|a - ab + b|U \setminus U'|\\
                       &\le |U|b + |S|a - ab,
\end{align*}
as desired.

Thus, we may now assume that $\deg_{{E}}(i) \ge b+1$ for all
$i \in U$.  {Observe that the condition $|E \cap (U \times S)| \le |U|b
+ |S|a - ab.$  is equivalent to.}

\[
  {\sum_{j \in S} (\deg_{E \cap (U \times [n])}(j) -a) \le |U|b - ab.}
\]

{Note that the RHS is independent of $S$, and each term on the
  LHS is an independent contribution for each $j \in S$. Therefore,
  the worst-case choice of $S$ is one of the following:}
{
\begin{enumerate}
\item $S$ is the set of all $j \in [n]$ with $\deg_{E \cap (U \times
    [n])}(j) > a$ \emph{if} this set has size at least $b$.
\item Otherwise, $S$ is the set of $b$ vertices $j \in [n]$ which
  are the largest with respect to $\deg_{E \cap (U \times [n])}(j)$.
\end{enumerate}
}

{In either case, to} apply the excess compatibility condition,
let $V \subset [n]$ be the $n-b$ vertices with \emph{lowest} degree
with respect to $E \cap (U \times [n])$. {Let $T$ be the set of
  vertices $j \in [n]$ such that
  $\deg_{{E \cap (U \times [n]})}(j) > a$. Observe that in case
  (1), we have that $S = T$ and in case (2), we have that
  $S = [n] \setminus V.$}

By the contrapositive of Lemma~\ref{lem:hall}, we have that
\begin{align*}
  \sum_{i \in U} e(i) &\le \sum_{j \in V} \min(a, \deg_{E \cap (U
                        \times V)}(j))\\
&{= \sum_{j \in V} \min(a, \deg_{E \cap (U
                        \times [n])}(j))}\\
&{= \sum_{j \in V \setminus T} \deg_{{E \cap (U \times
                                              [n]})}(j) + a |V
                                              \cap T|.}\\
&{= |E \cap (U \times (V \setminus T))| + a |V \cap T|.}
\end{align*}

Also note that
\begin{align*}
  \sum_{i \in U} e(i) &= \sum_{i \in U} (\deg_{{E}}(i) - b)\\
                      &= |E \cap (U \times [n])| - b|U|{.}
\end{align*}

 Combining the two equation blocks, we have that 
{
\begin{align*}
|E \cap (U \times S)| &= |E \cap (U \times [n])| - b|U| - |E \cap (U
                        \times ([n] \setminus S))| + b|U|\\
                      &= \sum_{i \in U} e(i) - |E \cap (U
                        \times ([n] \setminus S))| + b|U|\\
                     &\le |E \cap (U \times (V \setminus T))| + a|V
                       \cap T| - |E \cap (U \times ([n] \setminus S)| + b|U|.
\end{align*}}


{We split the remaining analysis into the two cases.
\begin{enumerate}
\item In this case, $S = T$. Since $|S| \ge b$ and consists of the
  largest degrees, we have that $|V
  \setminus T| = [n] \setminus S$ and $|S \setminus V| = b$. Therefore,
\begin{align*}
|E \cap (U \times S)| &\le |E \cap (U \times (V \setminus T))| + a|V
                       \cap T| - |E \cap (U \times (V \setminus
                       T))| + b|U|\\
                     &= a|V \cap S| + b|U|\\
                     &= a|S| + b|U| - a|S \setminus V|\\
                     &= a|S| + b|U| - ab.
\end{align*}

\item In this case, $S = [n] \setminus V$ and $|S| = b$. Therefore,
\begin{align*}
  |E \cap (U \times S)| &\le |E \cap (U \times (V \setminus T))| + a|V
                       \cap T| - |E \cap (U \times ([n] \setminus S)|
                          + b|U|\\
  &= a|V \cap T| - |E \cap (U \times (V \cup T))| + b|U|\\
  &= \sum_{j \in V \cap T} (a - \deg_{E \cap (U \times
       [n])} (j)) + b|U|\\
     &\le b|U|\\
     &= b|U| + ab - ab\\
     &= b|U| + a|S| - ab,
\end{align*}
where the fourth line follows from the definition of $T$. \qedhere
\end{enumerate}

}




%
%
%
%
%
%
%
%
%
%
\end{proof}

By {the regularity theorem of \cite{GHKSWY} (Theorem~\ref{Gop:Reg})}, regularity is equivalent to correctability of erasure patterns when $a=1$. Since Theorem~\ref{thm:regular-iff-excess} shows that regularity is equivalent to excess-compatibility, we have the equivalence between all three notions. %

\begin{corollary}\label{thm:excess-compat}
If $a=1$, an erasure pattern $E$ is excess-compatible iff $E$ is generically correctable.
\end{corollary}

We will now prove that excess-compatibility of an erasure pattern can be reduced to a max flow problem which can be solved in polynomial time. This implies that correctability of an erasure pattern by an $(m,n,a=1,b)$-MR tensor code can be checked in polynomial time. This proves Theorem~\ref{thm:fast}.

\begin{figure}

\begin{center}
\includegraphics[width=4in]{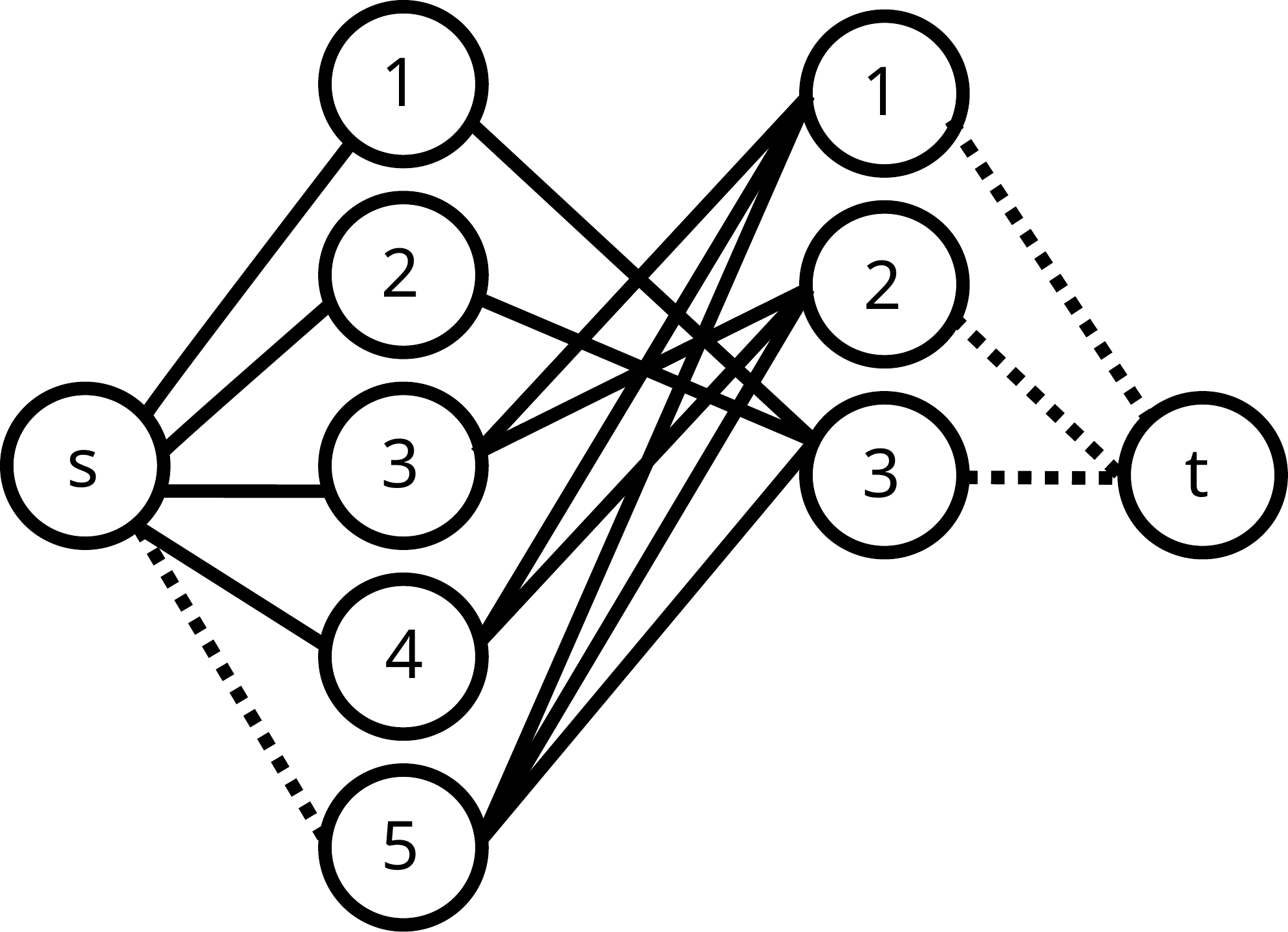}
\end{center}

\caption{{An example of the flow graph for $(m=5,n=5,a=2,b=2)$,
  $E =
  \{(1,3),(1,4),(1,5),(2,3),(2,4),(2,5),(3,1),(3,2),(3,5),(4,1),(4,2),(4,5),(5,1),(5,2),(5,3),(5,4)\}$
with $V = \{1,2,3\}$. Each solid edge has capacity $1$ and each dotted
edge has capacity $2$. In this case, a full-capacity flow exists.}}

\label{fig:flow}

\end{figure}

\begin{proposition}
Excess-compatibility of an erasure pattern is testable in time $$O(\min(\tbinom{m}{a},\tbinom{n}{b}) \cdot (m+n)^{3}).$$
\end{proposition}
\begin{proof}
We will prove that excess compatibility can be checked in the given time using a max flow algorithm.

Let $E \subset [m]\times [n]$ be the pattern we wish to test
excess-compatibility. Fix a subset $V \subset [n]$ of size
$\bar{b}=n-b$. Construct a directed graph on $m+\bar{b}+2$ nodes: a
source $s$, a sink $t$, $c_i$ for $i \in [m]$ and $d_j$ for $j \in
V$. For each $i \in [m]$, have a (directed) edge from $s$ to $c_i$
with capacity $e(i)$. For all $(i, j) \in E \cap ([m] \times V)$, have
an edge from $c_i$ to $d_j$ with capacity $1$. For all $j \in V$, have
an edge from $d_j$ to $t$ with capacity $a$. {See
  Figure~\ref{fig:flow} for an example.} Then, by definition of
$V$-excess flow, a $V$-excess flow $E' \subset E\cap ([m] \times V)$
exists if and only if the maximum flow in our constructed directed
graph saturates every edge from $s$ to $c_i$. This maximum flow can be
computed in $O(m+n)^3$ time {(e.g., the
  \textsf{Relabel-To-Front} algorithm in \cite{CLRS} runs in
  $O((\text{number of vertices})^3)$ time)}. Since we need to perform this check for all $V \subset [n]$ of size $\bar{b}$, the total running time is $O(\tbinom{n}{b} (m+n)^3).$ Since we could have also performed this test by swapping $m$ and $n$ (and $a$ and $b$), we can also test in $O(\binom{m}{a} (m+n)^3)$ time, as desired.
\end{proof}

\bibliographystyle{alpha}
\bibliography{references}

\appendix

\section{MR Tensor Code: Upper bound on field size}\label{app:upper-bound}

\begin{theorem}
There exists a $(m,n,a,b)$-MR Tensor Code with field size $O_{a,b,m}(n^{b(m-a)})$. 
\end{theorem}

\begin{proof}
  Let $C_{row} \otimes C_{col}$ be the code we seek to {construct}. We shall exhibit a system of equations, the sum of whose degrees is $O_{a,b,m}(n^{b(m-a)})$. By the Schwartz-Zippel Lemma \cite{DL78, Zip79, Sch80}, this will imply the existence of a code over a field size of $O_{a,b,m}(n^{b(m-a)})$.

  Let $H_{row}$ be the parity check matrix of $C_{row}$ of size $b\times n$, $H_{col}$ be the parity check matrix of $C_{col}$ of size $a\times m$. And let $$H = \begin{bmatrix}I_m \otimes H_{row}\\ H_{col}\otimes D_{(n-b)\times n}\end{bmatrix}$$ be the parity check matrix of $C_{row} \otimes C_{col}$. Here $I_m$ is the $m\times m$ identity matrix and $D$ is the $(n-b)\times n$ matrix formed by the first $n-b$ rows of $I_n$. Note that the number of columns of $H$ is $mn$ and the number of rows of $H$ is $mb+(n-b)a=mb+na-ab$ which is the codimension of the tensor code.

  Call an erasure pattern $E \subset [m] \times [n]$ \emph{minimal} if each nonempty row of $E$ has size least $b+1$ and each nonempty column of $E$ has size at least $a+1$.

  We impose the following three constraints on $H$,
  \begin{enumerate}
  \item $H_{row}$ is an MDS code.
  \item $H_{col}$ is an MDS code.
  \item For every $E \subset[m] \times [n]$ which is minimal and correctable, we impose that the minor $H|_{E}$ has rank $|E|$.
  \end{enumerate}

  First, we show these conditions are sufficient to ensure that $C_{row} \otimes C_{col}$ is an MR Tensor Code. Let $E \subseteq [m] \times [n]$ be any correctable pattern. Since $H_{row}$ is MDS, if any row of $E$ has at most $b$ entries, we can correct that row just by using $H_{row}$. Likewise, since $H_{col}$ is MDS, if and row of $E$ has at most $a$ entries, we can correct that row just by using $H_{col}$. By removing any such rows and columns iteratively, it suffices to correct some $E' \subseteq E$ which is minimal and correctable (as $E$ is correctable). The correctability of $E'$ follows from condition 3.

  To ensure the first condition, we sample the entries of $H_{row}$ randomly. To check MDS, we need to ensure that each $a \times a$ minor of the $[m, a]$ code has nonzero determinant. This is a system of equations of total degree $a\binom{m}{a} = O_{a,m}(1)$.

  Likewise, for the second condition, we can ensure $H_{col}$ is MDS with a system of equations of total degree $b\binom{n}{b} = O_b(n^{b}).$

  The last condition is a bit more tricky to analyze. First, we show that there are at most $O_{m,a,b}(n^{b(m-a)})$ minimal patterns, and each has $O_{m,a,b}(1)$ entries.

  Assume that a minimal correctable pattern $E$ has $u$ nonempty rows and $v$ nonempty columns. Because $E$ is correctable, it is regular and so $|E| \le ub + va - ab$. Also, since $E$ is minimal, we have that $|E| \ge v(a+1)$. Therefore,
  \[
    ub + va - ab \ge |E|\ge  v(a+1) \implies v\le b(u-a)\le b(m-a).
  \]
  Thus, $E$ spans at most $b(m-a)$ columns and has at most $ub + va - ab \le mb + b(m-a)a - ab = O_{a,b,m}(1)$ entries. Thus, the number of such $E$ is at most
  \[
    \binom{n}{\le b(m-a)} \cdot \binom{m\cdot b(m-a)}{\le mb+b(m-a)a-ab} = O_{a,b,m}(n^{b(m-a)}). 
  \]
  Ensuring that $H|_{E}$ has full rank is equivalent to some $|E|\times |E|$ minor of $H|_{E}$ having nonzero determinant. If $E$ is correctable, {then} one of these minors has at least one symbolically nonzero determinant. The constraint for this determinant has degree $|E| = O_{a,b,m}(1)$. Thus, we can specify all the necessary constraints with total degree $O_{a,b,m}(n^{b(m-a)})$.
\end{proof}

\section{{Properties of higher order MDS codes}}\label{app:mds-extra}

\subsection{Proof of Lemma~\ref{lemma:ultimate-mds-equiv}}

To prove this lemma, we start with some foundational claims.

\begin{claim}\label{claim:matrix-magic}
  For any $V \in \F^{k \times n}$ and for all $A_1, \hdots, A_\ell \subset [n]$,
  \[
    \dim(V_{A_1} \cap \cdots \cap V_{A_\ell}) = \sum_{i=1}^{\ell} \dim(V_{A_i}) - \rank \begin{bmatrix}V_{A_1} & V_{A_2} & & & \\
      V_{A_1} & & V_{A_3} & & \\
      \vdots & & & \ddots & &\\
    V_{A_1} & & & & & V_{A_\ell}\end{bmatrix}.
  \]
\end{claim}

\begin{remark}
This formula has previously appeared in the literature. For instance, see~\cite{tian2019formulas}.
\end{remark}

\begin{proof}
  We may assume that $\dim(V_{A_i}) = |A_i|$ for all $i$. Otherwise, some of the columns of $V_{A_i}$ are linear combinations of other columns. Thus, we can remove elements of $A_i$ corresponding to the redundant columns without changing the rank on either side of the main expression.

  In particular, the RHS is equal to the dimension of the kernel of
  \[
    X := \begin{bmatrix}V_{A_1} & V_{A_2} & & & \\
      V_{A_1} & & V_{A_3} & & \\
      \vdots & & & \ddots & &\\
    V_{A_1} & & & & & V_{A_\ell}\end{bmatrix}.
  \]
  It suffices to exhibit a linear bijection between $\ker X$ and $V_{A_1} \cap \cdots \cap V_{A_{\ell}}$. For any $x \in \ker {X}$, let $x^i$ be the entries of $x$ corresponding to $A_i$. In particular, we must have that
  \[
    -V_{A_1} x^1 = V_{A_2}x^2 = V_{A_3}x^3 = \cdots V_{A_\ell}x^\ell.
  \]
  Thus, $y := V_{A_2}x^2$ is in $V_{A_1} \cap \cdots V_{A_\ell}$. This map has an inverse. For any $y \in V_{A_1} \cap \cdots V_{A_\ell}$, there exists unique $x^1, \hdots, x^\ell$ (because $\dim(V_{A_i}) = |A_i|)$ such that $y = V_{A_1}x^1 = \cdots V_{A_\ell}x^{\ell}$. In that case $x = (-x^1, x^2, \hdots, x^{\ell})$ is in $\ker X$. This establishes the bijection.
\end{proof}

\begin{claim}\label{claim:generic-minimizes}
  Let $V \in \mathbb F^{k \times n}$ be MDS and $W \in \R^{k \times n}$ be generic. Then, for all $\ell \ge 2$ and $A_1, \hdots, A_{\ell} \subset [n],$
  \[
    \dim(V_{A_1} \cap \cdots \cap V_{A_{\ell}}) \ge \dim(W_{A_1} \cap \cdots \cap W_{A_{\ell}}). 
  \]
\end{claim}

\begin{proof}
  Since $V$ is MDS, $\dim(V_{A_i}) = \dim(W_{A_i})$ for all $i \in [\ell].$  Thus,
  \begin{align*}
    &\dim(V_{A_1} \cap \cdots \cap V_{A_{\ell}}) - \dim(W_{A_1} \cap \cdots \cap W_{A_{\ell}})\\
    &= \rank \begin{bmatrix}W_{A_1} & W_{A_2} & & & \\
      W_{A_1} & & W_{A_3} & & \\
      \vdots & & & \ddots & &\\
    W_{A_1} & & & & & W_{A_\ell}\end{bmatrix} - \rank \begin{bmatrix}V_{A_1} & V_{A_2} & & & \\
      V_{A_1} & & V_{A_3} & & \\
      \vdots & & & \ddots & &\\
    V_{A_1} & & & & & V_{A_\ell}\end{bmatrix},
  \end{align*}
  which is nonnegative because generic matrices maximize rank.
\end{proof}

\begin{claim}\label{claim:generic-removal}
  Let $V \in \mathbb \F^{k \times n}$ and $W \le \F^k$ be a subspace. For any $A \subseteq B \subseteq [n]$,
  \[
    0 \le \dim(W \cap V_{B}) - \dim(W \cap V_{A}) \le |B \setminus A|.
  \]
\end{claim}

\begin{proof}
  The left inequality is trivial. For the right, observe
  \begin{align*}  
    \dim(W \cap V_{B}) &- \dim(W \cap V_{A})\\
                       &= \dim(W) + \dim(V_B) - \dim(W + V_B) - \dim(W) - \dim(V_A) + \dim(W + V_A)\\
                       &= (\dim(V_B) - \dim(V_A)) + (\dim(W + V_A) - \dim(W + V_B))\\
                       &\le (\dim(V_A + V_{B\setminus A}) - \dim(V_A)) + (0)\\
                       &= \dim(V_{B\setminus A}) - \dim(V_A \cap V_{B\setminus A})\\
                       &\le |B \setminus A|.
  \end{align*}
\end{proof}

\begin{lemma}[Padding Lemma]\label{lem:generic-padding}
  Let $V \in \mathbb \F^{k \times n}$ be any MDS matrix. Let $\ell \ge 2$. Consider $A_1, \hdots, A_\ell \subset [n]$ of size at most $k$. Then the following statements are true:
  \begin{enumerate}
     \item If $\sum_{i=1}^\ell |A_i| > (\ell-1)k$, then $V_{A_1} \cap \cdots \cap V_{A_{\ell}}\ne 0.$
     \item If $\sum_{i=1}^\ell |A_i| \le (\ell-1)k$ and $n$ is sufficiently large, then $V_{A_1} \cap \cdots \cap V_{A_{\ell}} = 0$ iff there exist $A_1' \supseteq A_1, \hdots, A_\ell' \supseteq A_\ell$ such that
     \begin{enumerate}
     \item $A_1'\setminus A_1, A_2'\setminus A_2,\dots,  A_\ell'\setminus A_\ell$ and $A_1\cup A_2\cup \dots \cup A_\ell$ are mutually disjoint,
     \item $|A_i'| \le k$ for all $i \in [\ell]$. 
     \item  $|A_1'| + \cdots + |A_{\ell}'| = (\ell-1)k$,
     \item $V_{A'_1} \cap \cdots \cap V_{A'_{\ell}} = 0$.
      \end{enumerate} 
  \end{enumerate}
\end{lemma}

\begin{proof}

  \begin{enumerate}
  \item By Claim~\ref{claim:matrix-magic}, we know there exists a matrix $M$ of $(\ell-1)k$ rows and $\sum_{i=1}^{\ell} |A_i|$ columns such that
    \begin{align*}
      \dim(V_{A_1} \cap \cdots \cap V_{A_{\ell}}) &= \sum_{i=1}^\ell \dim(V_{A_i}) - \rank(M)\\
                                             &\ge \sum_{i=1}^{\ell} |A_i| - (\ell-1)k\\
      & > 0,
    \end{align*}
    Thus, $V_{A_1} \cap \cdots \cap V_{A_{\ell}} \neq 0$.

  \item The `if' direction follows from $V_{A_1} \cap \cdots \cap V_{A_\ell} \subseteq V_{A'_1} \cap \cdots \cap V_{A'_\ell} = 0$.

    For the `only if' direction, assume $n \ge (\ell-1)k$. By Claim~\ref{claim:matrix-magic},

    \[0 = \dim(V_{A_1} \cap \cdots \cap V_{A_\ell})= \sum_{i=1}^{\ell} \dim(V_{A_i}) - \rank \begin{bmatrix}V_{A_1} & V_{A_2} & & & \\
        V_{A_1} & & V_{A_3} & & \\
        \vdots & & & \ddots & &\\
        V_{A_1} & & & & & V_{A_\ell}\end{bmatrix}.
    \]
    Let $U(A_1, \hdots, A_{\ell}) \subset \mathbb F^{(\ell-1)k}$ be the column space of the block matrix in the above expression. We let $U_i \subset U$ be the subspace of $U$ which is supported on the $i$th block of $k$ coordinates of $F^{(\ell-1)k}$. Note that $U_1 \oplus U_2 \oplus \cdots \oplus U_{\ell-1} \subseteq U$. Therefore, $U = \mathbb F^{(\ell-1)k}$ if and only if $\dim U_i = k$ for all $i \in [\ell-1]$.

    We now `grow' $A_1, \hdots, A_{\ell}$ into the desired $A'_1, \hdots, A'_{\ell}$ through the following inductive process. Let $t = (\ell-1)k - \sum_{i=1}^{\ell} |A_i|$. Let $A^{(0}_1, \hdots, A^{(0)}_{\ell}$ be $A_1, \hdots, A_{\ell}$. 

    \begin{itemize}
    \item For $i \in \{0, \hdots, t-1\}$.
    \item Since $\sum_{a=1}^{\ell} |A_a| < (\ell-1)k$, we have $U(A_1, \hdots, A_{\ell}) \neq \mathbb F^{(\ell-1)k}$. Thus, we can identify $j \in [\ell-1]$ such that $\dim(U_j) < k$.
    \item Add an element in $[n] \setminus \bigcup_{a=1}^\ell A^{(i)}_a$ to $A^{(i)}_{j+1}$. Call this new family $A^{(i+1)}_1, \hdots, A^{(i+1)}_\ell$. 
    \item Repeat these steps.
    \end{itemize}

    We let $A'_i = A^{(t)}_i$ for all $i \in [\ell]$. Clearly $A'_1, \hdots, A'_\ell$ are supersets of $A_1, \hdots A_\ell$. We claim they also satisfy properties (a)-(d). Property (a) is satisfied because a new element of $[n]$ is added at each step of the algorithm. Property (b) is satisfied because if no elements are added to $A_1$ and for all $j \ge 2$ an element is added to $A^{(i)}_j$ only if $\dim(U_j) < k$, which cannot happen if $|A^{(i)}_j| = k$; therefore no set will exceed $k$ in size at any point. Property (c) is satisfied because the algorithm runs for $t = (\ell-1)k - \sum_{i=1}^{\ell} |A_i|$ steps.

    For property $d$, we claim by induction for all $i \in \{0, 1, \hdots, t\}$, $V_{A^{(i)}_1} \cap \cdots \cap V_{A^{(i)}_\ell} = 0$. The base case $i=0$ follows by assumption.
    Note that at each stage, $\dim U(A^{(i+1)}_1, \hdots, A^{(i+1)}_\ell) \ge \dim U(A^{(i)}_1, \hdots, A^{(i)}_\ell) +1$ because one of the subspaces $U_j$ increases in dimension by padding $A_{j+1}$. Thus,
    \begin{align*}
      \dim(V_{A^{(i+1)}_1} \cap \cdots \cap V_{A^{(i+1)}_\ell})&= \sum_{j=1}^{\ell} \dim(V_{A^{(i+1)}_j}) - \dim (U(A^{(i+1)}_1, \hdots, A^{(i+1)}_\ell))\\
                                                               &\le \left(1+\sum_{j=1}^{\ell}\dim (V_{A^{(i)}_j})\right) - (\dim (U(A^{(i)}_1, \hdots, A^{(i)}_\ell)) +1)\\
                                                               &= 0.
      \end{align*}

      Thus, $V_{A^{(i)}_1} \cap \cdots \cap V_{A^{(i)}_\ell} = 0$, completing the induction.
  \end{enumerate}
  
\end{proof}

Now we prove Lemma~\ref{lemma:ultimate-mds-equiv}.

\lemmaone*

\begin{proof}
  Observe that ``only if'' direction follows immediately. Now we seek to show that ``if'' direction. We do this by showing the contrapositive.

  Fix $A_1, \hdots, A_{\ell}$ such that $\dim(V_{A_1} \cap V_{A_2} \cap \cdots \cap V_{A_\ell}) \neq \dim(W_{A_1} \cap \cdots \cap W_{A_{\ell}}),$ Because $\dim(W_{A_1} \cap \cdots \cap W_{A_{\ell}})$ is the rank of a generic matrix, we must have that

  \[\dim(V_{A_1} \cap V_{A_2} \cap \cdots \cap V_{A_\ell}) > \dim(W_{A_1} \cap \cdots \cap W_{A_{\ell}}) =: d,\]
  Note that $|A_i| \ge d$ for all $i \in [\ell]$. We claim there exists subsets $A'_1 \subset A_1, A'_2 \subseteq A_2, \hdots, A'_{\ell} \subseteq A_{\ell}$ with $|A'_1| + |A'_2| + \cdots + |A'_{\ell}| = d$ such that
  \[W_{A_1\setminus A'_1} \cap W_{A_2 \setminus A'_2} \cap \cdots \cap W_{A_{\ell} \setminus A'_{\ell}} = 0.\]
  This follows from Claim~\ref{claim:matrix-magic}, as the block matrix has rank $d$ less than the number of columns, so $d$ columns can be removed without changing the rank, which decreases the dimension of the intersection by $d$. Observe that we must have
  \[\dim(V_{A_1 \setminus A'_{1}} \cap V_{A_2 \setminus A'_2} \cap \cdots \cap V_{A_\ell \setminus A'_{\ell}}) > 0\]
  because the dimension can decrease by at most $d$.

  Let $B_i = A_i \setminus A'_i$ for all $i \in [\ell]$. If any $|B_i| > k$, then $V_{B_i} = \mathbb F^k$ and $W_{B_i} = \mathbb R^k$. We can replace both with an arbitrary subset of size $k$ without changing anything. Thus, $\sum_i |B_i| \le \ell k$. In fact since
  \[
    0 = \dim(W_{B_1} \cap \cdots \cap W_{B_{\ell}}) \le k - \sum_{i} (k - |B_i|). 
  \]
  we have that $\sum_{i} |B_i|\le (\ell-1)k.$ Now, if the inequality is strict, we can add an element to one of the $B_i$s without changing that the generic intersection of the $W$s is nonempty (this is by looking at the matrix view and noting that some row must be less than full rank). this can only increase the intersection of the $V$s. This finishes the argument that $\neg 2 \Rightarrow \neg 3.$ 
  
\end{proof}

\subsection{Proof of Lemma~\ref{lem:correctability_intersection}}

\lemmatwo*
\begin{proof}
  Note that by Proposition~\ref{prop:correctability}, correctability is equivalent to (\ref{eq:2}).
  Observe that (\ref{eq:2})$\iff$(\ref{eq:3}) because the addition of vector spaces is commutative.

  We now show that (\ref{eq:2})$\iff$(\ref{eq:4}). By symmetric argument, we may show that (\ref{eq:3})$\iff$(\ref{eq:5}).

   The conditions on $E$ translate to the following conditions on $A_1,A_2,\dots, A_m.$
   \begin{enumerate}
   \item $|A_i|\le n-b$
   \item $\sum_i |A_i|= (m-a)(n-b).$ 
   \end{enumerate}
   The following statements are equivalent.
   \begin{enumerate}
   \item $E$ is not correctable.
   \item There exists a non-zero codeword of $C$ which is supported on $E.$ 
   \item There exist $r_1,r_2,\dots,r_m\in C_\row$, not all zero, such that 
     \begin{itemize}
     \item $\supp(r_i)\subset \barA_i$ for $i\in [m]$,
     \item $$P\cdot \begin{bmatrix}
       r_1\\
       r_2\\
       \vdots\\
       r_m
     \end{bmatrix}=\sum_{i=1}^m P_i \otimes r_i^T = 0.$$
     \end{itemize}
     (Since $r_i=y_i^T V$ for some $y_i\in \F^\barb$ and $I \otimes V^T$ is an injective linear map, we have the following equivalent statement.)
   \item There exist $y_1,y_2,\dots,y_m \in \F^{\barb}$, not all zero, such that 
     \begin{itemize}
     \item $y_i^T V_{A_i} =0$ for $i\in [m]$,
     \item $\sum_{i=1}^m P_i \otimes y_i =0.$
     \end{itemize}
     (Since $y_i \in V_{A_i}^\perp$ for each $i\in [m],$ we have the following equivalent statement.)
   \item There exists $y_i\in V_{A_i}^\perp$, not all zero, such that $\sum_{i=1}^m P_i \otimes y_i =0$.\ \ 
     (Since $|A_i|\le n-b$, and $V$ is a generator matrix of an MDS code, we have $\dim(V_{A_i}^\perp)=(n-b)-\dim(V_{A_i})=(n-b)-|A_i|$. Therefore $$\sum_{i=1}^m \dim(V_{A_i}^\perp) = \sum_{i=1}^m (n-b-|A_i|)=m(n-b)-(m-a)(n-b)=a(n-b)=a\bar{b}.$$ Since $P_i$ is one-dimensional, we also have $\sum_{i=1}^m \dim(P_i \otimes V_{A_i}^\perp)=a\bar{b}.$ So we have the following equivalent statement.)
   \item $P_1 \otimes V_{A_1}^\perp + P_2 \otimes V_{A_2}^\perp + \dots + P_m \otimes V_{A_m}^\perp \ne \F^a \otimes \F^{\barb}.$
   \end{enumerate}
   
   This completes the proof.
   
 \end{proof}

\subsection{Counterexample to $\MDS(\ell)$ duality for $\ell \ge 4$}

Consider the following matrices over $\F_{13}$.
\begin{align*}
  V &= \begin{pmatrix}
    1 & 1 & 1 & 1 & 1 & 1 & 1 &1\\
    0 & 1 & 2 & 3 & 4 & 5 & 6 & 7
  \end{pmatrix}\\
  V^{\perp} &= \begin{pmatrix}
    1 & 0 & 0 & 0 & 0 & 0 & 6 & 6\\
    0 & 1 & 0 & 0 &0 & 0 & 7 & 5\\
    0 & 0 & 1 & 0 & 0 & 0 & 8 & 4\\
    0 & 0 & 0 & 1 & 0 & 0 & 9 & 3\\
    0 & 0 & 0 & 0 & 1 & 0 & 10 & 2\\
    0 & 0 & 0 & 0 & 0 & 1 & 11 & 1
  \end{pmatrix}                                   
\end{align*}
One can check that $V$ is MDS(4) (this follows from $V$ being MDS) but $V^{\perp}$ is not MDS. For example, let $A_1 = \{1, 2, 3, 4, 5\}$, $A_2 = \{1,2,3,6,7\}$, $A_3 = \{1,2,4,6,8\}$, and $A_4 = \{5,7,8\}$. If $V^{\perp}$ were a generic matrix, then  $V^{\perp}_{A_1} \cap V^{\perp}_{A_2} \cap V^{\perp}_{A_3} \cap V^{\perp}_{A_4} = 0$, but one can verify that $(5,4,3,2,8,0) \in V^{\perp}_{A_1} \cap V^{\perp}_{A_2} \cap V^{\perp}_{A_3} \cap V^{\perp}_{A_4}$.

A computer search can also find a number of other counterexamples, such as $\MDS(8,4,4)$ codes whose duals are not $\MDS(4)$. Another way to see the failure of duality of $\MDS(4)$ is by the lower bound of \cite{kong2021new} where they proved that $(m=4,n,a=1,b=2)$-MR tensor codes require fields of size $\Omega(n^2)$. Constructing $(m=4,n,a=1,b=2)$-MR tensor codes is equivalent to constructing $(n,n-2)-\MDS(4)$ code by Theorem~\ref{thm:MRtensorcodesa1_MDSm}. If $\MDS(4)$ duality is true, this is equivalent to {constructing} $(n,2)-\MDS(4)$ codes. It is easy to see any $(n,2)$-MDS code is also an $\MDS(4)$ code. Since there are MDS codes over linear size fields (Reed-Solomon codes), this would violate the lower bound of \cite{kong2021new}.

\subsection{Proof of Proposition~\ref{prop:reduction}}

\propone*
\begin{proof}

  (1)   Follows trivially by taking $A_\ell$ in (\ref{eqn:MDSell}) to be the entire set.\\
  (2) Let $V_{k\times n}$ be a generator matrix for $C$. Then dropping the $i^{th}$ column of $V$, we get the generator matrix for $C_1$, the puncturing of $C$ at $i$. Therefore the condition (\ref{eqn:MDSell}) still holds.\\
  (3) WLOG, let's assume that the code is shortened at position $n$. Let $V_{k\times n}=[V_1 V_2 \cdots V_n]$ be the generator matrix of $C$. By a basis change, we can assume that $V_n=e_k$, the $k^{th}$ coordinate vector in $\F^k$. Let $\tV_i\in \F^k$ be the vector formed by dropping the last coordinate of $V_i$. It is easy to see that $\tV=[\tV_1 \tV_2 \cdots \tV_{n-1}]$ is the generator matrix for the shortened code $C_1.$ We now want to prove that $C_1$ is $\MDS(\ell).$ 
  Let $\tA_1,\tA_2,\dots,\tA_{\ell}\subset [n-1]$ with $\sum_{i=1}^\ell |\tA_i|=(k-1)(\ell-1)$ such that $\bigcap_{i=1}^\ell \tW_{\tA_i}=0$ for a generic $(k-1)\times n$ matrix $\tW$. By Lemma~\ref{lemma:ultimate-mds-equiv}, it is enough to show that $\bigcap_{i=1}^\ell \tV_{\tA_i}=0$. Define $A_1=\tA_1$ and $A_i=\tA_i \cup \{n\}$ for $i\ge 2.$ 
  \begin{claim}
    \label{claim:Wintersection_shortening}
    For $W_{k\times n}$ and $\tW_{(k-1)\times n}$ are generic matrices, then $$\bigcap_{i=1}^\ell \tW_{\tA_i}=0 \iff \bigcap_{i=1}^\ell W_{A_i}=0.$$
  \end{claim}
  \begin{proof}
    WLOG, by a basis change we can assume $W_n=e_k$, the $k^{th}$ coordinate vector. Let $\tW_i$ to be the vector formed by dropping the last coordinate of $W_i$, clearly $\tW$ is also generic. Now
      $\bigcap_{i=1}^\ell W_{A_i}=0$ iff
    \begin{align*}
      X=\begin{bmatrix}W_{A_1} & W_{A_2} & & & \\
      W_{A_1} & & W_{A_3} & & \\
      \vdots & & & \ddots & &\\
    W_{A_1} & & & & & W_{A_\ell}\end{bmatrix}
    \end{align*}
    is full rank. Doing column operations, we can conclude that $X$ is full rank iff
        \begin{align*}
      \tX=\begin{bmatrix}\tW_{\tA_1} & W_{\tA_2} & & & \\
      W_{\tA_1} & & W_{\tA_3} & & \\
      \vdots & & & \ddots & &\\
    W_{\tA_1} & & & & & W_{\tA_\ell}\end{bmatrix}
    \end{align*}
    is full rank. This is equivalent to $\bigcap_{i=1}^\ell \tW_{\tA_i}=0$.
  \end{proof}
  Since $V$ is the generator for an $\MDS(\ell)$ code, we have $\bigcap_{i=1}^\ell W_{A_i}=0 \Rightarrow \bigcap_{i=1}^\ell V_{A_i}=0$. So therefore we will be by proving that $\bigcap_{i=1}^\ell V_{A_i}=0 \Rightarrow \bigcap_{i=1}^\ell \tV_{\tA_i}=0$, the proof of which is essentially same as that of Claim~\ref{claim:Wintersection_shortening}.
  \end{proof}

\section{Near constructions of $(n,3)$-$\MDS(3)$ codes}\label{app:constructions}

In this appendix, we exhibit some partial progress toward constructing $(n,3)-\MDS(3)$ codes. By Theorem~\ref{thm:main-lower-bound}, we know such a construction needs $\Omega(n^2)$ field size. We construct a couple of different codes with field size $O(n^2)$ {which} have some of the properties of $\MDS(3)$ codes.

The constructions are inspired by the fact that Reed-Solomon codes produce the (nearly) optimal field size for $\MDS$ codes.

\subsection{A weak bipartite $\MDS(3)$ construction}

Let $p$ be a prime, and $q=p^2$. Assume that $\mathbb F_q = \mathbb F_p[X]/\langle p(X)\rangle$, where $p(X)$ is a degree-2 irreducible in $\mathbb F_p[X]$.

\begin{lemma}\label{lem:bipartite-stuff}
  There are explicit $u_0,u_1, \hdots, u_{p-1} \in \F_q^3$ and $v_0,v_1, \hdots, v_{p-1} \in \F_q^3$ with the following property. Let $W_{\alpha,\beta} = \operatorname{span}(u_\alpha v_\beta)$. Then $W_{\alpha_1,\beta_1} \cap W_{\alpha_2,\beta_2} \cap W_{\alpha_3,\beta_3} \neq 0$ if they generically should.%
\end{lemma}

We call it ``weak bipartite'' as the spaces we are considering the {intersection} of form a bipartite graph.

\begin{remark}
Note that the space $W_{1,1}$ has $q^2 - 1$ nonzero points, and the $q-1$ nonzero points of $W_{i,j} \cap W_{1,1}$ must be unique for all $i,j \in [2, n]$. Therefore, any bipartite MDS(3) construction must have $O(\sqrt{q})$ points.  Therefore, this construction is essentially optimal.
\end{remark}

\begin{remark}
Warning! In the construction the $u_i$'s and $v_i$'s are \emph{not} MDS. In fact, they are collinear. Even so, this seems to be one of the few known algebraic constructions which gives a generic 3-wise intersection condition.
\end{remark}

\begin{proof}
  For all $\alpha,\beta \in \F_p$, let $u_\alpha = (\alpha, -1, 0)$ and $v_\beta = (\beta + \beta^2X, 0, -1)$.

  Now assume for arbitrary $\alpha_1,\alpha_2,\alpha_3,\beta_1,\beta_2,\beta_3 \in \F_p$ that $W_{\alpha_1,\beta_1} \cap W_{\alpha_2,\beta_2} \cap W_{\alpha_3,\beta_3} \neq 0$. Let $w_{\alpha,\beta} = u_{\alpha} \times v_{\beta} = (1, \alpha, \beta+\beta^2X)$. Then, we must have that $w_{\alpha_1,\beta_1}, w_{\alpha_2,\beta_2}, w_{\alpha_3,\beta_3}$ are coplanar. In other words
  \[
    \det \begin{bmatrix}1 & 1 & 1\\
      \alpha_1 & \alpha_2 & \alpha_3\\
      \beta_1+\beta_1^2X & \beta_2 + \beta_2^2X & \beta_3+\beta_3^2X
    \end{bmatrix} = 0.
  \]
  Subtracting the first column from the second and third columns and then expanding we have that
  \[
    0 = \det\begin{bmatrix}\alpha_2-\alpha_1 & \alpha_3-\alpha_1 \\ (\beta_2-\beta_1)(1+(\beta_1+\beta_2)X) & (\beta_3-\beta_1)(1+(\beta_1+\beta_3)X)\end{bmatrix}.
  \]
  Expanding, we get that
  \[
    (\alpha_2 - \alpha_1)(\beta_3 - \beta_1)(1 + (\beta_1+\beta_2)X) = (\alpha_3-\alpha_1)(\beta_2-\beta_1)(1+(\beta_1+\beta_3)X).
  \]
  Comparing coefficients of powers of $X$, we get that $\eta := (\alpha_2-\alpha_1)(\beta_3-\beta_1) = (\alpha_3-\alpha_1)(\beta_2-\beta_1)$ and $\eta \beta_2 = \eta \beta_3$.

  If $\eta = 0$, then either $\alpha_1=\alpha_2=\alpha_3$, $\beta_1 = \beta_2 = \beta_3$, or $(\alpha_1,\beta_1) = (\alpha_i, \beta_i)$ for $i = 2$ or $3$. In each of these cases, the spaces $W_{\alpha_1,\beta_1}, W_{\alpha_2,\beta_2}, W_{\alpha_3,\beta_3}$ generically intersect.

  If $\eta \neq 0$ then $\beta_2 = \beta_3$ which then either $\beta_1=\beta_2=\beta_3$ or $\alpha_2 = \alpha_3$. Again, in each of these cases, the spaces $W_{\alpha_1,\beta_1}, W_{\alpha_2,\beta_2}, W_{\alpha_3,\beta_3}$ generically intersect.

  Thus, our construction is bipartite MDS(3).
\end{proof}

\subsection{A very weak $\MDS(3)$ Construction}

Note that the lower bound for the field size of $\MDS(\ell)$, only needed the following structure: that there is \emph{one} {partition} of the vectors into $\ell$ groups, such that $k-k/\ell$-dimensional subspaces, one drawn from each group, has trivial intersection. Let's call this property \emph{very weak} $\MDS(\ell)$. We now show that the lower bound is essentially tight if $k=\ell=3$.

We use a coset trick of \cite{gopi2020maximally}. Let $p$ be a prime such that $p \equiv 1 \mod 3$. Thus, there exists $\zeta \in \mathbb F_p^*$ which is a nontrivial cube root of $1$. Let $S \subset \F_p^*$ be a subgroup of size $(p-1)/3$ not containing $\zeta$. Let $c \in \mathbb F_p$ be a non-quadratic residue, and consider the field extension $\mathbb F_q := \F_p[X] / \langle X^2 - c\rangle$, where $q = p^2$. Now define three sets of vectors
\begin{align*}
  U &= \{(1, \alpha, \alpha^2) : \alpha \in S\}\\
  V &= \{(1, \zeta \beta, \zeta^2 \beta^2) : \beta \in S\}\\
  W &= \{(1, X\gamma, X^2\gamma^2) : \gamma \in \{1, 2, \hdots, (p-1)/2\}\}.
\end{align*}
These three sets are disjoint, and their union is a subset of the Reed-Solomon code over $\mathbb F_q^3$, so $U \cup V \cup W$ is MDS(2). 

Now we seek to show that for any $u_1, u_2 \in U$, $v_1, v_2 \in V$, $w_1, w_2 \in W$ that $\Span (u_1, u_2) \cap \Span (v_1, v_2) \cap \Span (w_1, w_2) = 0$. Have the notation $u_i := (1, \alpha_i, \alpha_i^2)$, etc. We know that $\Span(u_1, u_2)^{\perp} = \Span((\alpha_1\alpha_2, -(\alpha_1 + \alpha_2), 1)),$ etc. Thus, $\Span (u_1, u_2) \cap \Span (v_1, v_2) \cap \Span (w_1, w_2) \neq 0$ if and only if
\[
  \det \begin{bmatrix}
    \alpha_1\alpha_2 & \zeta^2\beta_1\beta_2 & c\gamma_1\gamma_2\\
    \alpha_1+\alpha_2 & \zeta(\beta_1 + \beta_2) & X(\gamma_1 + \gamma_2)\\
    1 & 1 & 1
  \end{bmatrix} = 0.
\]

In particular, this implies that the coefficient of $X$ in the expansion of the determinant is equal to $0$. That is,
\[
  (\gamma_1+\gamma_2)\left[\alpha_1\alpha_2 - \zeta^2\beta_1\beta_2 \right] = 0.
\]
But, $\gamma_1+\gamma_2 \in \{1, \hdots, p-1\}$ and $\frac{\alpha_1\alpha_2}{\beta_1\beta_2} \in S \not\ni \zeta^2$. Thus, the above expression cannot be $0$, contradiction. Therefore, for all $u_1, u_2 \in U$, $v_1, v_2 \in V$, $w_1, w_2 \in W$ that $\Span (u_1, u_2) \cap \Span (v_1, v_2) \cap \Span (w_1, w_2) = 0$, as desired.

Observe that $|U|, |V|, |W| \ge (p-1)/3 = \Omega(\sqrt{q})$, which is tight up to constant factors by the proof of Lemma~\ref{lem:weak-MDS-lowerbound}.

\end{document}